\documentclass[12pt,a4paper]{article}
\RequirePackage[OT1]{fontenc}
\RequirePackage{amsthm,amsmath}
\RequirePackage[numbers]{natbib}
\RequirePackage[colorlinks,citecolor=blue,urlcolor=blue]{hyperref}
\usepackage[textsize=footnotesize]{todonotes}
\usepackage{hyperref}
\usepackage{cite}

\usepackage{graphicx}     
\usepackage{amsmath}
\usepackage{amssymb}
\usepackage{amsthm} 
\usepackage[countmax]{subfloat}
\usepackage{fourier}
\usepackage{xcolor} 
\usepackage{mathtools}
\usepackage{enumerate}
\usepackage{natbib}
\usepackage{url}
\usepackage{verbatim}
\usepackage{booktabs}
\usepackage{placeins}
\usepackage{multirow}
\usepackage{caption}
\usepackage{subcaption}
\usepackage{mathrsfs}
\def\P{\mathbb P}
\def\R{\mathbb R}
\def\E{\mathbb E}

\newtheorem{prop}{Proposition}[section]
\newtheorem{thm}{Theorem}[section]
\newtheorem{lemma}{Lemma}[section]

\newtheorem{Remark}{Remark}[section]
\newtheorem{Corollary}{Corollary}[section]

\newcommand{\Cn}{\widehat{C}_N}

\allowdisplaybreaks

\begin{document}

\title{A test for separability in covariance operators of random surfaces }
\author{Pramita Bagchi\\
George Mason University,\\
Department of Statistics\\
Fairfax, VA -22030 \\
USA\\
\\
Holger Dette \\
Ruhr-Universit\"at Bochum \\
Fakult\"at f\"ur Mathematik \\
44780 Bochum \\
Germany}
\date{}
\maketitle

\begin{abstract}
The assumption of separability is a simplifying and very popular
 assumption in the analysis of spatio-temporal  or  hypersurface data structures. It is often
 made in situations where the covariance structure cannot be easily estimated, for example because
 of a small sample size or because of computational storage problems.
 In this paper we propose a new and very simple test to validate this assumption. 
 Our approach is based on a  measure of  separability
which is zero in the case of separability and positive otherwise. 
{We derive 
the asymptotic distribution  of a  corresponding
estimate  under the null hypothesis and the alternative
and develop  an asymptotic and a bootstrap test, which are very easy to implement. In particular, our approach does neither require projections on
 subspaces generated by the eigenfunctions of the covariance operator nor distributional assumptions as recently used by \citet{Aston2017} and \citet{conkokrei2017}
to construct tests for separability.} We investigate the finite sample performance by means of a  simulation  study and also provide a comparison with the currently available methodology. Finally,
 the new procedure is illustrated analyzing a data example.
\end{abstract}

\section{Introduction}\label{sec1}
\def\theequation{1.\arabic{equation}}
\setcounter{equation}{0}

Data, which is functional {\bf and} multidimensional is usually called surface data 
 and  arises in areas such as  medical imaging [see  \citet{worsley1996,skup2010longitudinal}],  spectrograms derived from audio signals
  or geolocalized data  [see \citet{rabiner1978,bar2008spatio}].
  In many of these ultra high-dimensional  problems a completely non-parametric estimation of the covariance
  operator is not possible as  the number of available observations   is  small compared to the dimension.
  A common approach to obtain reasonable estimates in this context  are structural assumptions on the
   covariance  of the underlying process, and in recent years the assumption of
   separability has become  very popular, for  example in the analysis of geostatistical space-time models
   [see \citet{genton2007,gnegengut2007}, among others]. Roughly speaking, this assumption allows to write the covariance $$c(s,t,s^\prime, t^\prime) =
   \E [X(s,t) X(s^\prime,t^\prime) ] $$ of a (real valued) space-time process $\{X(s,t) \}_{(s,t) \in  S\times T}$
   as a product of the space  and
   time covariance function, that is
\begin{equation}\label{1.1}
c(s,t,s^\prime, t^\prime) = c_1(s,s^\prime) c_2(t, t^\prime).
\end{equation}
It was pointed  out by many authors that the   assumption  of separability yields a substantial simplification of the estimation problem
 and thus reduces  computational  costs in  the estimation of the covariance in high dimensional problems [see
 for example  \citet{huizengaetal2002,rougier2017representation}].
  Despite of its importance, there exist only a few  tools to validate 
 the  assumption of separability for surface data.

Many authors  developed tests for spatio-temporal data. For example,  
\citet{fuentes2006} proposed a test based on the spectral representation,
  and   \citet{luzimm2005,mitgengum2005,mitgengum2006} investigated likelihood ratio tests under the assumption of a normal distribution.
 Recently,  \citet{conkokrei2017}
   derived the joint distribution of the three statistics appearing in the likelihood ratio test and used this result to derive the asymptotic
   distribution of the (log) likelihood ratio. These authors also proposed alternative  tests which are based on distances between an estimator of the covariance
   under the assumption of separability and an estimator which does not use this assumption. 
   
    \citet{Aston2017} considered the problem of testing for separability in the context of hypersurface data. These authors pointed out  that many available methods require  the estimation of the full multidimensional covariance structure, which
    can become infeasible for high dimensional data. In order to address this issue they developed
    a bootstrap test for  applications, where replicates from the underlying random process
    are available.  To avoid estimation and  storage of the full covariance
   finite-dimensional projections of the difference between the covariance operator
    and a nonparametric separable approximation  {(using the partial trace operator)} were  proposed.
    In particular they suggested to  project onto subspaces  generated by the eigenfunctions of the covariance operator estimated under the
assumption of separability.  However, as pointed in the same references the choice of the number of eigenfunctions onto which one should project is not trivial and  the test might be sensitive
with respect  to this  choice. Moreover, the computational costs increase substantially with the number of eigenfunctions.

In this paper we present an alternative and  simple test for the hypothesis of separability in hypersurface data. We consider a similar
setup as in    \citet{Aston2017} and proceed in two steps. 
 {
Roughly speaking we derive an {\it explicit} expression for the minimal distance between the covariance operator
and its approximation by a separable covariance operator.
It turns out that this minimum vanishes if and only if the covariance operator is separable.}
Secondly, we directly estimate the minimal distance (and not the covariance operator itself) from the available data. 
{ As a consequence the
calculation of the test statistic does neither  use an estimate of the full non-separable covariance operator nor requires the specification
of subspaces used for a projection.}

  In Section \ref{sec2} we  review some basic terminology and 
   {discuss the problem of finding a best approximation of  the covariance operator by  a  separable covariance operator.
 The corresponding  minimum distance could also be
 interpreted as a measure of deviation from separability (it is zero in the case of separability and positive otherwise). 
 In Section \ref{sec3} we propose an estimator
 of the minimum distance, prove its consistency  and derive its asymptotic distribution  under the null hypothesis
 and alternative. These results are also used to develop an asymptotic and a  bootstrap test for the hypothesis of separability, which are - in contrast to the currently
 available methods - consistent against all alternatives.
  Moreover, statistical guarantees can be derived under more general and easier to verify moment assumptions 
 than  in \citet{Aston2017}.} 
 Section \ref{sec4} is devoted to an investigation of the finite sample properties of the new tests and a comparison with  two alternative tests for this problem, which have recently been proposed by   \citet{Aston2017} and  \citet{conkokrei2017}. In particular
 we demonstrate that - despite of their simplicity - the new procedures 
 have very competitive properties compared to the currently available methodology.
 Finally, some technical details are deferred to the Appendix  \ref{sec5}.

\section{Hilbert spaces and a measure of separability} \label{sec2}
\def\theequation{2.\arabic{equation}}
\setcounter{equation}{0}
\hspace{-0.1 in} We begin introducing some basic facts about Hilbert spaces, Hilbert-Schmidt operators and tensor products. For more details we refer to the monographs of  \citet{weidmann1980},  \citet{dunford1988} or \citet{gohberg1990}.
Let $H$ be a real separable Hilbert space with inner product $\langle  \cdot , \cdot \rangle$ and norm $\|  \cdot \|$. The space of bounded linear operators on $H$ is denoted by $S_{\infty}(H)$ with operator norm
$$
\VERT T\VERT_{\infty} := \sup_{\|f\| \leq 1}\|Tf\|.
$$
A bounded linear operator $T$ is said to be compact if it can be written as 
$$T=\sum_{j\ge1}s_j(T)\langle e_j,\cdot\rangle f_j,
$$
where $\{e_j:j\ge1\}$ and $\{f_j:j\ge1\}$ are orthonormal sets of $H$, $\{s_j(T):j\ge1\}$ are the singular values of $T$ and the series converges in the operator norm. We say that a compact operator $T$ belongs to the Schatten class of order $p\ge1$ and write $T \in S_p(H)$ if
$$
\VERT T \VERT_p= \Big (\sum_{j\ge1}{ { s_j(T)^p}} \Big )^{1/p} <\infty.
$$
 The Schatten class of order $p\ge1$ is a Banach space with  norm
$\VERT . \VERT_p$ and with the property that $S_p(H) \subset S_q(H)$ for $p < q$.  In particular we are interested in Schatten classes of order  $p=1$ and $2$. A compact operator $T$ is called Hilbert-Schmidt operator if $T\in S_2( H)$ and trace class if $T\in S_1( H)$. The space of Hilbert-Schmidt operators $S_2( H)$ is also a Hilbert space with the Hilbert-Schmidt inner product given by
$$
\langle A,B\rangle
=\sum_{j\ge1}\langle Ae_j,Be_j\rangle
$$
for each $A,B\in S_2(H)$, where $\{e_j:j\ge1\}$ is an orthonormal basis (the inner product does not depend on the choice of the basis).

{ Let $H_1$ and $H_2$ be two real separable Hilbert spaces. For $u \in H_1$ and $v \in H_2$, we define the bilinear form $u \otimes v: H_1 \times H_2 \to \mathbb{R}$ by
$$[u \otimes v](x,y) := \langle u,x \rangle \langle v,y \rangle, ~~~~ (x,y) \in H_1 \times H_2.$$
Let $\mathcal{P}$ be the set of all finite linear combinations of such bilinear forms. An inner product on $\mathcal{P}$ can be defined as 
$\langle u \otimes v, w \otimes z \rangle = \langle u,w \rangle \langle v,z \rangle$, for $u,w \in H_1$ and $v,z \in H_2$. The completion of $\mathcal{P}$ under the aforementioned inner product is called the tensor product of $H_1$ and $H_2$ and denoted as $H_1 \otimes H_2$.}

For $C_1 \in S_{\infty}(H_1)$ and $C_2 \in S_{\infty}(H_2)$, the tensor product $C_1 \tilde{\otimes} C_2$ is defined as the unique linear operator on $H := H_1 \otimes H_2$ satisfying
$$(C_1 \tilde{\otimes} C_2)(u \otimes v) = C_1u \otimes C_2v, ~~ u \in H_1, v \in H_2.$$
In fact $C_1 \tilde{\otimes} C_2 \in S_{\infty}(H)$ with $\VERT C_1 \tilde{\otimes} C_2\VERT_{\infty} = \VERT C_1 \VERT_{\infty} \VERT C_2 \VERT_{\infty}$. Moreover, if $C_1 \in S_p(H_1)$ and $C_2 \in S_p(H_2)$, for $p \geq 1$, then $C_1 \tilde{\otimes} C_2 \in S_{p}(H)$ with $\VERT C_1 \tilde{\otimes} C_2\VERT_{p} = \VERT C_1 \VERT_{p} \VERT C_2 \VERT_{p}$. For more details we refer to Chapter 8 of ~\citet{weidmann1980}. In the sequel we denote the Hilbert-Schmidt inner product on $S_2(H)$ with $H = H_1 \otimes H_2$ as $\langle \cdot ,\cdot\rangle_{HS}$ and that of $S_2(H_1)$ and $S_2(H_2)$ as $\langle \cdot,\cdot \rangle_{S_2(H_1)}$ and $\langle \cdot,\cdot \rangle_{S_2(H_2)}$ respectively.

\subsection{Measuring separability} \label{sec21}
We consider random elements $X$ in the Hilbert space $H$ with $\E\|X\|^4< \infty$  {[see Chapter 7 in \citet{hsingeubank2015} for more details on 
Hilbert space valued random variables].} Then the covariance operator of $X$ is defined as
$C := \E\left[(X-\E X)\otimes_{o} (X- \E X)\right]$, where 
for $ f,g\in H$ the operator $f \otimes_{o} g: H \to H $ is defined by
$$
(f \otimes_o g)h = \langle h,g \rangle f  \ \mbox {for all} \ h \in H.
$$
Under the assumption $\E\|X\|^4< \infty$, we have $C \in S_2(H)$. We also assume $\VERT C\VERT_2 \neq 0$, which essentially means the random variable $X$ is non-degenerate. 
To test separability we consider the hypothesis
\begin{equation} \label{H0A}
H_0 : C = C_1 {\widetilde{\otimes}} C_2   \  \mbox{for some} \  C_1 \in S_2(H_1)  \  \mbox{and} \ C_2 \in S_2(H_2).
\end{equation}
Our approach is based on an approximation of the operator $C$ by a separable operator $C_1 \tilde{\otimes} C_2$
with respect to the norm $\VERT \cdot \VERT_2$. 
 Ideally,
we are looking for
\begin{equation} \label{mindist}
D:= \inf \Big \{ \VERT C - C_1 \tilde{\otimes} C_2\VERT_2^2 ~\Big|
~ C_1 \in S_2(H_1), C_2 \in  S_2(H_2)\Big \}~,
\end{equation}
such that the hypothesis of separability in \eqref{H0A} can be rewritten in terms of the distance $D$, that is
\begin{equation} \label{h0equiv}
H_0 : D=0 ~~ \mbox{versus} ~~H_1 : D >  0~.
\end{equation}
It turns out that it is difficult to express $D$ explicitly in terms of the covariance operator $C$. For this reason we proceed in a slightly different way in two steps. 
 {First we fix an operator $C_1^* \in S_2(H_1)$  and
determine
\begin{equation} \label{mindistc1}
D_{C_1^*} := \inf \Big \{ \VERT C - C_1^* \tilde{\otimes} C_2\VERT_2^2 ~\Big|
~ C_2 \in  S_2(H_2)\Big \}~,
\end{equation}
that is we are minimizing
$\VERT C - C_1^* \tilde{\otimes} C_2\VERT_2^2$ with respect to second factor $C_2$ of the tensor product. 
In particular  we will  show that the infimum is in fact a minimum and derive an explicit expression for $D_{C_1^*}$ and its minimizer.
Instead of working with the distance $D$ in \eqref{mindist} we suggest to  estimate  an appropriate distance from the family 
$$\{ D_{C_1^*} | C_1^* \in  S_1(H_1) \}.
$$
For this purpose note that 
for a  given covariance operator $C \in S_2(H)$ and 
$C_1^* \in S_2(H_1)$  the corresponding distance $D_{C_1^*}$ is in general positive. 
However, we also show in the following 
that  $C$ is separable, i.e. 
$C = C_1 \widetilde{\otimes} C_2 $,  if and only if 
the corresponding  minimum  $D_{C_1}$  vanishes. 
Thus, if we are able to estimate $D_{C_1}$ (for the unknown operator $C_1$),  we can 
test the hypothesis \eqref{h0equiv}, by constructing a test for the hypotheses 
$H_0:D_{C_1} = 0 ~\mbox{ versus } ~H_1:D_{C_1} > 0. $
We explain below that this is in fact possible.} 
\\
 For this purpose we have to introduce some additional notation and have to prove several auxiliary results. The main statement is given in Theorem \ref{lem:dist} (whose formulation also requires the new notation). First, consider the  bounded linear operator
 $T_1:S_2(H) \times S_2(H_1) \mapsto S_2(H_2)$     defined by
\begin{equation} \label{t1def} T_1(A \tilde{\otimes} B,C_1) = \langle A,C_1\rangle_{S_2(H_1)} B
\end{equation}
 for all $C_1 \in S_2(H_1)$. Similarly, let $T_2: S_2(H) \times S_2(H_2) \to S_2(H_1)$ be the bounded linear operator defined by
 \begin{equation} \label{t2def} T_2(A \tilde{\otimes} B,C_2) = \langle B,C_2\rangle_{S_2(H_2)}A
\end{equation}
for all $C_2 \in S_2(H_2)$.
 {The proof of the following two auxiliary results can be found in Section \ref{a2} and \ref{a3} of the Appendix.} 

\begin{prop}
\label{prop:t1}
The operators $T_1$ and $T_2$ are
 well-defined, bi-linear and continuous with 
\begin{eqnarray}
\label{T1}
\langle B, T_1(C,C_1) \rangle_{S_2(H_2)} = \langle C, C_1 \tilde{\otimes} B  \rangle_{HS},\\
\label{T2}
\langle A, T_2(C,C_2) \rangle_{S_2(H_1)} = \langle C, A\tilde{\otimes} C_2  \rangle_{HS}.
\end{eqnarray}
for all $A, { C_1} \in S_2(H_1)$, $B, { C_2} \in S_2(H_2)$ { and $C \in S_2(H)$}. 
\end{prop}
 {
While the previous result clarifies the existence of the operators $T_1$ and $T_2$, the next proposition provides a property of the operator $T_1$, which is essential for the proof of the main  result of this section.
}
%
%

\begin{prop}
\label{prop:t1alt}
For any $C \in S_2(H)$ and $C_1 \in S_2(H_1)$, we have
$$\langle C, C_1 \tilde{\otimes} T_1(C,C_1)\rangle_{HS} = \VERT T_1(C,C_1)\VERT_2^2 .$$
\end{prop}

\begin{thm}
\label{lem:dist}
For each $C \in S_2(H)$ and any fixed $C_1^* \in S_2(H_1)$ the distance 
\begin{equation}\label{minop1}
D_{C_1^*}(C_2) =\VERT C-C_1^* \tilde{\otimes} C_2 \VERT_2
\end{equation}
is minimized at
\begin{equation}\label{minop}
\widetilde{C}_2 = \frac{T_1(C,C_1^*)}{\VERT C_1^* \VERT_2^2}.
\end{equation}
Moreover, the minimum distance in \eqref{minop1} is given by
\begin{equation} \label{mindist2}
 D_{C_1^*} = \VERT C \VERT_2^2 - \frac{\VERT T_1(C,C_1^*)\VERT^2_2}{\VERT C_1^* \VERT_2^2}.
\end{equation}
In particular  $D_{C_1^*}$ is zero if and only if $C = C_1^* \tilde{\otimes} C_2$ for some 
$C_2 \in S_2(H_2)$.

\end{thm}

\begin{proof}
We write
\begin{align*}
D_{C_1^*}(C_2) = \VERT C - C_1^* \tilde{\otimes} C_2\VERT_2^2 
= &\VERT C - C_1^* \tilde{\otimes} \widetilde{C}_2 \VERT_2^2 + \VERT C_1^* \tilde{\otimes} \widetilde{C}_2 -C_1^* \tilde{\otimes} C_2\VERT_2^2 \\
&+2\langle C - C_1^* \tilde{\otimes} \widetilde{C}_2 ,  C_1^* \tilde{\otimes} \widetilde{C}_2 -C_1^* \tilde{\otimes} C_2\rangle_{HS}.
\end{align*}
For the last term we obtain from \eqref{minop}
\begin{align*}
\langle C - C_1^* \tilde{\otimes} \widetilde{C}_2 ,  C_1^* \tilde{\otimes} \widetilde{C}_2 -C_1^* \tilde{\otimes} C_2\rangle_{HS}
= &\langle C, C_1^* \tilde{\otimes} \widetilde{C}_2\rangle_{HS} - \VERT C_1^* \tilde{\otimes} \widetilde{C}_2\VERT_2^2 \\
& - \langle C, C_1^* \tilde{\otimes} C_2 \rangle_{HS} + \langle C_1^* \tilde{\otimes}  \widetilde{C}_2, C_1^* \tilde{\otimes} C_2\rangle_{HS}\\
= &\frac{1}{\VERT C_1^* \VERT_2^2}\langle C,  C_1^* \tilde{\otimes} T_1(C,C_1^*)\rangle_{HS} - \frac{\VERT T_1(C,C_1^*)\VERT_2^2}{\VERT C_1^* \VERT_2^2} \\
&- \langle C, C_1^* \tilde{\otimes} C_2 \rangle_{HS} + \langle C_2,T_1(C,C_1^*) \rangle_{HS},
\end{align*}
which  is zero by \eqref{T1} and Proposition \ref{prop:t1alt}. Therefore we have
 for all $C_2 \in S_2(H_2)$
$$\VERT C - C_1^* \tilde{\otimes} C_2\VERT_2^2 \geq \VERT C - C_1^* \tilde{\otimes} \widetilde{C}_2\VERT_2^2 $$
which proves the  first assertion of  Theorem \ref{lem:dist}.

For a proof of the representation \eqref{mindist2} we substitute 
 the operator $\widetilde{C}_2$ defined in \eqref{minop} for $C_2$ in the expression of $D_{C_1^*}(C_2)$ and obtain
\begin{align*}
D_{C_1^*} = &
D_{C_1^*}(\widetilde{C}_2) =  \VERT C - C_1^* \tilde{\otimes} \widetilde{C}_2 \VERT_2^2 
=  \langle C - C_1^* \tilde{\otimes} \widetilde{C}_2, C - C_1^* \tilde{\otimes} \widetilde{C}_2\rangle_{HS}\\
= & \VERT C \VERT_2^2 + \VERT C_1^* \tilde{\otimes} \widetilde{C}_2\VERT_2^2 - 2\langle C, C_1^* \tilde{\otimes} \widetilde{C}_2\rangle_{HS}\\
= & \VERT C \VERT_2^2 + \frac{\VERT T_1(C,C_1^*) \VERT_2^2}{\VERT C_1^* \VERT_2^2} - \frac{2}{\VERT C_1^* \VERT_2^2}\langle C, C_1^* \tilde{\otimes} T_1(C,C_1^*)\rangle_{HS}.
\end{align*}
Then the second assertion follows from Proposition \ref{prop:t1alt}. 

For the last part, now assume that 
$C = C_1^* \tilde{\otimes} C_2$ for some $C_2 \in S(H_2)$, then \eqref{t1def} implies
$$
\VERT T_1 (C, C_1^*)\VERT^2_2 = ( \langle C_1^*, C_1^* \rangle_{S_2(H_1)})^2 \VERT  C_2 \VERT^2_2 =\VERT  C_1^* \VERT_2^4\VERT  C_2 \VERT^2_2
$$
{
and observing the representation \eqref{mindist2}
we obtain  $D_{C_1^*} =0$. Conversely, if $D_{C_1^*} = 0$, we have $\VERT C - C_1^* \tilde{\otimes} \widetilde{C}_2 \VERT_2= 0$ where $\widetilde{C}_2  = {T_1(C,C_1^*)}/{\VERT C_1 \VERT_2^2}
$. Consequently $C = C_1^* \tilde{\otimes} \widetilde{C}_2$. 
}
\end{proof}

{\begin{Remark} \label{rem21}
{\rm 
By  Theorem \ref{lem:dist} we can construct a test for the hypothesis 
$$
H_0: D_{C_1^*} = 0
$$
for any given $C_1^* \in S_2(H_1)$
by estimating the quantity in \eqref{mindist2}. 
If the covariance operator $C$ is not separable it follows that $D_{C_1^*} > 0$ for all
$C_1^* \in S_2 (H_1)$. If $C$ is 
in  fact separable (i.e., the null hypothesis is true) such that   $C = C_1 \widetilde{\otimes} C_2$ for some $C_1$ and $C_2$ we have $D_{C_1}=D_{C_1}(C_2)=0$.  
Interestingly we can 
obtain  $C_1$ from $C$ up to a multiplicative constant
using the operator $T_2$ defined in \eqref{t2def}.
More precisely, let  $\Delta$ be any fixed element of $S_2(H_2)$ and    note that under the null hypothesis of separability we have $T_2(C,\Delta) = \langle C_2,\Delta \rangle_{S_2(H_2)}C_1$. As the minimum distance in \eqref{mindist2} is 
invariant with respect to scalar multiplication of $C_1^*$ it follows  for this choice
\begin{equation} \label{mindistopt}
D_0 := D_{C_1}=  D_{T_2(C,\Delta)} =\VERT C \VERT_2^2 - \frac{\VERT T_1(C,T_2(C,\Delta))\VERT^2_2}{\VERT T_2(C,\Delta) \VERT_2^2}.
\end{equation}
Note that $D_0\geq 0 $ and  $D_0$ vanishes if and only if $C$ is separable. Thus we can construct a consistent test of the hypothesis \eqref{H0A} via a suitable estimate of the operator $C$ in \eqref{mindistopt}. This program is carefully carried out in Section \ref{sec3}.
}
\end{Remark}

\begin{Remark} \label{rem22}
{\rm 
Note that the representation \eqref{mindistopt} involves only norms  of operators and as a consequence, when it comes to estimation,  we do not have to store the complete estimate of the covariance kernel. We make this point more precise in Remark \ref{rem4}, where we discuss  the problem of estimating $D_0$ in the context of Hilbert-Schmidt integral operators.
}
\end{Remark}
}

\subsection{Hilbert-Schmidt integral operators}\label{sec22} 

 An important  case for applications 
is the set $H := L^2\left(S \times T\right)$ of all real-valued square integrable functions defined on  $S \times T$, where $S \subset \R^p$, $T\subset \R^q$ are bounded measurable sets. 
If the covariance operator $C$ of the random variable $X$  is a Hilbert-Schmidt operator  it follows from Theorem 6.11 in \citet{weidmann1980} that  there exists a kernel
$c\in  L^2 \big ( (S \times T) \times (S \times T) \big )$
such that $C$  can be characterized 
as an integral operator, i.e.
$$Cf(s,t) = \int_S \int_T c(s,t,s',t')f(s',t')ds'dt', ~~~
f \in L^2 (S \times T) ,
$$
almost everywhere on $ S \times T$. Moreover the kernel 
is given by the covariance kernel of $X$, that is  $c(s,t,s',t') = \text{Cov}\left[X(s,t),X(s',t')\right]$, and  the space $H$ can be identified with the tensor product of $H_1 = L^2(S)$ and $H_2 = L^2(T)$.

Similarly, if $C_1$ and $C_2$ are Hilbert-Schmidt operators on $L^2(S)$ and $L^2(T)$ respectively  there exists symmetric kernels $c_1 \in L^2(S \times S)$ and $c_2 \in L^2(T \times T)$ such that,
\begin{eqnarray*}
C_1f(s) &=& \int_S c_1(s,s')f(s')ds' ,~ f \in H_1,\\
C_2g(t) &=& \int_T c_2(t,t')g(t')dt' ,~g \in H_2
\end{eqnarray*}
almost everywhere on $S$ and $T$, respectively.   The following result shows that in this case the operators $T_1$ and $T_2$ defined  by \eqref{t1def} and \eqref{t2def}, respectively, are also {Hilbert-Schmidt  
integral} operators.
 {The proof can be found in Section \ref{a4} of the Appendix
and requires that the sets $S$ and $T$ are bounded.}

\medskip
\begin{prop}
\label{prop:t1kernel}
If $C$ and $C_1$ are integral operators with {continuous} kernels $c \in L^2\big ( (S \times T)  \times (S \times T)\big )$ and $c_1 \in L^2(S \times S)$, then $T_1(C,C_1)$ is an integral operator with kernel given by
\begin{equation}\label{kdef}
  k(t,t') = \int_S \int_{S} c(s,t,s't')c_1(s,s')ds ds'.
\end{equation}
An analog result is true for the operator $T_2$.
\end{prop}

\medskip

Using the explicit formula for $T_1$ described in Proposition  \ref{prop:t1kernel} the minimum distance can be expressed in terms of the corresponding kernels of the operators, that is
\begin{eqnarray}
\nonumber
D_{C_1}  &=& 
\VERT C \VERT_2^2 - \frac{\VERT T_1(C,C_1)\VERT^2_2}{\VERT C_1 \VERT_2^2} \\
&=& \int_T\int_T\int_S\int_S c^2(s,t,s',t')dsds'dtdt'  
\label{eq:dker}\\
&&- \frac{\int_T\int_T\left[\int_S\int_S c(s,t,s't')c_1(s,s')dsds'\right]^2dtdt'}{\int_S\int_S c_1^2(s,s')dsds'}.
\nonumber
\end{eqnarray}

\section{Estimation and asymptotic  properties} \label{sec3}
\def\theequation{3.\arabic{equation}}
\setcounter{equation}{0}

{
Formally we estimate the minimum distance given in \eqref{mindist2} by plugging in estimators for $C$ and $C_1$ based on a sample $X_1, X_2, \dots, X_N$. The covariance operator $C$ is estimated by
\begin{equation}
\label{estgen}
\widehat{C}_N := \frac{1}{N}\sum_{i=1}^{N}\left[(X_i - \overline{X})\otimes_{o}(X_i - \overline{X})\right].
\end{equation}
}
{As pointed out in Remark \ref{rem21}
it is sufficient  to estimate the operator $C_1$ up to a multiplicative constant, due to the self-normalizing form of the second term of the minimum distance $D_{C_1} $. Let $\Delta$ be any fixed element of $S_2(H_2)$, 
recall  that under the null hypothesis of separability
$H_0: C = C_1 \otimes C_2$ we have $T_2(C,\Delta) = \langle C_2,\Delta \rangle_{S_2(H_2)}C_1$. 
Observing the representation \eqref{mindistopt}
we suggest to estimate (a multiple of) the operator $C_1$  by \begin{equation}
\label{estgenc1}
\widehat{C}_{1N} = T_2(\widehat{C}_N,\Delta).
\end{equation}
With this choice we  obtain the test statistic 
\begin{equation} \label{estDN}
\widehat D_N = \VERT \widehat{C}_N \VERT_2^2  - \frac{\VERT T_1(\widehat{C}_N,T_2(\widehat{C}_N,\Delta))\VERT^2_2}{\VERT T_2(\widehat{C}_N,\Delta) \VERT_2^2} .
\end{equation}
As this representation only involves norms 
 we do not have to store the complete estimate of the covariance kernel (see   Remark \ref{rem4} for a more detailed discussion of this property).

\subsection{Weak convergence} \label{sec31}
The main results of this section provide the asymptotic properties of the statistic $\widehat D_N$
under the null hypothesis of separability and the alternative.

\begin{thm}
\label{thm:null_dist}
Assume that ~$\E\| X \|_2^4 < \infty$  and the null hypothesis of separability holds. Then  we have 
\begin{align}
\nonumber
N\widehat{D}_N \stackrel{d}{\to} 
&\left\VERT \mathcal G - \frac{T_2(\mathcal G,\Delta)\widetilde{\otimes}T_1(C,T_2(C,\Delta))}{\VERT T_2(C,\Delta)\VERT_2^2} \right\VERT_2^2 - \frac{\VERT T_1(\mathcal G,T_2(C,\Delta)) - T_1(C,T_2(\mathcal G,\Delta)\VERT_2^2}{\VERT T_2(C,\Delta)\VERT_2^2} \\
\label{eq:null_dist}
&   =
\left\VERT \mathcal G - \frac{T_2(\mathcal G,\Delta)\widetilde{\otimes}C_2 }{ \langle C_2 ,\Delta \rangle_{S_2(H_2)}  } \right\VERT_2^2 
- 
\left \VERT
 \frac{ T_1(\mathcal G,C_1)}{\VERT C_1 \VERT_2 } -
\frac{
\langle C_1 ,
T_2(\mathcal G,\Delta) \rangle_{S_2(H_1)} C_2}{ \langle C_2 ,\Delta \rangle_{S_2(H_2)} \VERT C_1 \VERT_2 }
\right \VERT_2^2,
\end{align} 
where
~$\mathcal{G}$ is a centered Gaussian process with covariance operator 
\begin{equation}\label{limop}
\Gamma := \lim_{N \to \infty}\mbox{Var}(\sqrt{N}\Cn)
\end{equation}
\end{thm}
\begin{proof}
The equality in \eqref{eq:null_dist} follows by a direct calculation using  \eqref{t1def} and  \eqref{t1def}. 
For the proof of the first part define the mapping $\phi:S_2(H) \mapsto \mathbb{R}$  by 
\begin{equation*}
\phi(A) = \VERT A \VERT_2^2 {\VERT T_2(A,\Delta)\VERT_2^2} - {\VERT T_1(A,T_2(A,\Delta))\VERT_2^2}.
\end{equation*}
By Proposition 5 in \citet{dauxois1982asymptotic}  the random variable $\sqrt{N}(\Cn - C)$ converges in distribution to a centered Gaussian random element $\mathcal{G}$  with variance 
\eqref{limop}
 in $S_2(H)$ with respect to Hilbert-Schmidt topology
 and we will first derive  the asymptotic distribution of $\phi(\widehat{C}_N) - \phi(C)$
using the functional delta method as described in Section 20.1 in \citet{van2000asymptotic}. For this purpose
we determine the derivatives of the map
$\phi_{C,G}: t \mapsto \phi(C+tG)$ for any fixed $G \in S_2(H)$.
Note that $\phi_{C,G} (t) $ is a polynomial in $t$. More precisely, we  have
\begin{eqnarray*}
 \phi(C+tG) &= &\VERT C+tG \VERT_2^2 {\VERT T_2(C+tG,\Delta)\VERT_2^2} - {\VERT T_1(C+tG,T_2(C+tG,\Delta))\VERT_2^2} \\
 &=& (a_0 + a_1t + a_2t^2)({c_0+c_1t+c_2t^2}) - ({b_0 + b_1t + b_2t^2 + b_3t^3 + b_4t^4}),
\end{eqnarray*}
where
\begin{eqnarray*}
a_0 &=& \VERT C \VERT_2^2~,~~a_1=2\left\langle C , G \right\rangle_{HS}~,~~ a_2 = \VERT G \VERT_2^2 ~\\
c_0 &=& 
\VERT T_2(C,\Delta)\VERT_2^2 ~,~~ c_1= 2\left\langle T_2(C,\Delta),T_2(G,\Delta)\right\rangle_{S_2(H_1)} ~,~~c_2=
\VERT T_2(G,\Delta)\VERT_2^2~
\end{eqnarray*}
and 
\begin{eqnarray*}
b_0 &=& \VERT T_1(C,T_2(C,\Delta)) \VERT_2^2~,~~ 
b_4 = \VERT T_1(G,T_2(G,\Delta)) \VERT_2^2 \\
b_1 &=&  2\big [\left\langle T_1(C,T_2(C,\Delta)),T_1(C,T_2(G,\Delta))\right\rangle_{S_2(H_2)} \\
&&~~~~
+ \left\langle T_1(C,T_2(C,\Delta)),T_1(G,T_2(C,\Delta))\right\rangle_{S_2(H_2)}\big ]\\
b_2 &=& \left[2\left\langle T_1(C,T_2(C,\Delta)),T_1(G,T_2(G,\Delta))\right\rangle_{S_2(H_2)} + \VERT T_1(C,T_2(G,\Delta)) \VERT_2^2\right.\\
&&~~~~\left.+ \VERT T_1(G,T_2(C,\Delta)) \VERT_2^2 + 2 \left\langle T_1(G,T_2(C,\Delta)),T_1(C,T_2(G,\Delta))\right\rangle_{S_2(H_2)}\right]\\
b_3 &=&2 \big [\left\langle T_1(G,T_2(G,\Delta)),T_1(C,T_2(G,\Delta))\right\rangle_{S_2(H_2)} \\
&&~~~~ + \left\langle T_1(G,T_2(G,\Delta)),T_1(G,T_2(C,\Delta))\right\rangle_{S_2(H_2)}\big]
\end{eqnarray*}
Now, under the null hypothesis of separability we have   for the quantity  in \eqref{mindistopt}  $D_{C_1}=0$ and $T_2(C,\Delta) = \langle C_2,\Delta \rangle_{S_2(H_2)}C_1$, which implies
\begin{eqnarray} \label{term0}
\phi( C + tG) |_{t=0} =  \phi (C) =
a_0c_0-b_0=0.
\end{eqnarray}
Similarly, using the fact that $C=C_1\widetilde{\otimes}C_2$  and $\frac{\VERT T_1(C,T_2(C,\Delta)) \VERT_2^2}{\VERT T_2(C,\Delta)\VERT_2^2}= \VERT C \VERT_2^2$  it  follows that
\begin{align*}
& {\left\langle T_1(C,T_2(C,\Delta)),T_1(C,T_2(G,\Delta))\right\rangle_{S_2(H_2)}}
={\langle C_2,\Delta \rangle_{S_2(H_2)} \left\langle T_1(C,C_1),T_1(C,T_2(G,\Delta))\right\rangle_{S_2(H_2)}} \\
& ~~~~~~~~~= {\langle C_2,\Delta \rangle_{S_2(H_2)} \langle C_1, C_1 \rangle_{S_2(H_1)} \left\langle C_2,T_1(C,T_2(G,\Delta))\right\rangle_{S_2(H_2)}}\\
& ~~~~~~~~~={\langle C_2,\Delta \rangle_{S_2(H_2)} \langle C_1, C_1 \rangle_{S_2(H_1)} \left\langle C_2,C_2\right\rangle_{S_2(H_2)} \langle C_1,T_2(G,\Delta) \rangle_{S_2(H_1)} }
\\
& ~~~~~~~~~=\VERT C \VERT_2^2 {\langle C_2,\Delta \rangle_{S_2(H_2)} \langle C_1,T_2(G,\Delta) \rangle_{S_2(H_1)} }
= \VERT C \VERT_2^2 { \langle T_2(C,\Delta),T_2(G,\Delta) \rangle_{S_2(H_1)} }
\end{align*}
and
\begin{align*}
  & {\left\langle T_1(C,T_2(C,\Delta)),T_1(G,T_2(C,\Delta))\right\rangle_{S_2(H_2)}}
=  {\langle C_2 ,\Delta \rangle_{S_2(H_2)}^2 \left\langle T_1(C,C_1),T_1(G,C_1)\right\rangle_{S_2(H_2)}}\\
& ~~~~~~~~~={\langle C_1,C_1 \rangle_{S_2(H_1)} \langle C_2,T_1(G,C_1)\rangle_{S_2(H_2)} \langle C_2 ,\Delta \rangle_{S_2(H_2)}^2}
= \langle C_2,T_1(G,C_1)\rangle_{S_2(H_2)} {\VERT T_2(C,\Delta)\VERT_2^2}  \\
& ~~~~~~~~~= \langle G, C_1 \widetilde{\otimes} C_2\rangle_{HS} 
{\VERT T_2(C,\Delta)\VERT_2^2} = \langle G, C \rangle_{HS} 
{\VERT T_2(C,\Delta)\VERT_2^2},
\end{align*}
which implies 
\begin{eqnarray} \label{term1}
\frac{d}{dt}\phi( C + tG) |_{t=0} = 
a_1c_0+a_0c_1-b_1= 0
\end{eqnarray}
(under the null hypothesis). 
Next we look at the second derivative and note 
the identities
\begin{align*}
 \VERT C \VERT_2^2{\VERT T_2(G,\Delta)\VERT_2^2}
 = & {\langle C_1,C_1 \rangle_{S_2(H_1)} \langle C_2,C_2\rangle_{S_2(H_2)} \langle T_2(G,\Delta),T_2(G,\Delta) \rangle_{S_2(H_1)}} \\
= &{\langle C_1,C_1 \rangle_{S_2(H_1)}  \langle T_2(G,\Delta)\widetilde{\otimes} C_2,T_2(G,\Delta) \widetilde{\otimes} C_2 \rangle_{HS}}
\\
= & 
\frac{\langle C_2,\Delta \rangle_{S_2(H_2)}^2 \langle C_1,C_1 \rangle_{S_2(H_1)}^2  \langle T_2(G,\Delta)\widetilde{\otimes} C_2,T_2(G,\Delta) \widetilde{\otimes} C_2 \rangle_{HS}}{\VERT T_2(C,\Delta)\VERT_2^2}\\
= & \frac{\left \langle T_2(G,\Delta)\widetilde{\otimes}T_1(C,T_2(C,\Delta)), T_2(G,\Delta)\widetilde{\otimes}T_1(C,T_2(C,\Delta)) \right\rangle_{HS}}{\VERT T_2(C,\Delta)\VERT_2^2}.
\end{align*}
\begin{align*}
    \left\langle T_1(C,T_2(G,\Delta)),T_1(G,T_2(C,\Delta))\right\rangle_{S_2(H_2)}  = &  \left\langle C_2 \langle C_1,T_2(G,\Delta)\rangle_{S_2(H_1)} ,T_1(G,C_1) \langle C_2,\Delta \rangle_{S_2(H_2)} \right\rangle_{S_2(H_2)}\\
    = & \langle C_1,T_2(G,\Delta)\rangle_{S_2(H_1)} \langle C_2,\Delta \rangle_{S_2(H_2)} \left\langle C_2  ,T_1(G,C_1)  \right\rangle_{S_2(H_2)}\\
    = & \left\langle \langle C_2,\Delta \rangle_{S_2(H_2)} C_1,T_2(G,\Delta)\right\rangle_{S_2(H_1)} \left\langle C_1 \widetilde{\otimes}C_2  ,G  \right\rangle_{HS}\\
    = & \left\langle T_2(C,\Delta) ,T_2(G,\Delta)  \right\rangle_{S_2(H_1)}  \left\langle C,G \right\rangle_{HS}
\end{align*}
and 
\begin{align*}
 &{\left\langle T_1(C,T_2(C,\Delta)),T_1(G,T_2(G,\Delta))\right\rangle_{S_2(H_2)}}
= {\langle C_2,\Delta \rangle_{S_2(H_2)} \langle C_1,C_1 \rangle_{S_2(H_1)} \langle C_2, T_1(G,T_2(G,\Delta))\rangle_{S_2(H_2)}} \\
& ~~~~~~~~~={\langle C_2,\Delta \rangle_{S_2(H_2)} \langle C_1,C_1 \rangle_{S_2(H_1)} \langle G, T_2(G,\Delta)\widetilde{\otimes}C_2\rangle_{HS}}\\
& ~~~~~~~~~= {\Big \langle G, T_2(G,\Delta)\widetilde{\otimes}\big(\langle C_2,\Delta \rangle_{S_2(H_2)} \langle C_1,C_1 \rangle_{S_2(H_1)} C_2\big) \Big\rangle_{HS}}\\
& ~~~~~~~~~={\left \langle G, T_2(G,\Delta)\widetilde{\otimes}T_1(C,T_2(C,\Delta)) \right\rangle_{HS}}~,
\end{align*}
(here we make extensive use of Proposition \ref{prop:t1}). 
This gives (observing the definition of $a_0,a_1,a_2,c_0,c_1,c_2,b_2$)
\begin{eqnarray}
\frac12
\frac{d^2}{d^2t}\phi( C + tG) |_{t=0} &= &
 a_0 c_2+ a_1 c_1 + a_2 c_0 -b_2  \\
 \nonumber
&=&  \left\VERT G {\VERT T_2(C,\Delta)\VERT_2}- \frac{T_2(G,\Delta)\widetilde{\otimes}T_1(C,T_2(C,\Delta))}{\VERT T_2(C,\Delta)\VERT_2} \right\VERT_2^2  \\
&-& {\VERT T_1(G,T_2(C,\Delta)) - T_1(C,T_2(G,\Delta)\VERT_2^2}.
\nonumber
%
\end{eqnarray}
Finally, taking $G := \sqrt{N}(\widehat{C}_N - C)$,  $t = 1/\sqrt{N}$ and using a Taylor expansion, we obtain (note that $\phi (C) =0$ under the null hypothesis)
\begin{eqnarray*}
N \phi(\widehat{C}_N)  &=&  \frac12
\frac{d^2}{d^2t}\phi( C + t\sqrt{N}(\widehat{C}_N - C))) |_{t=0} +o_p(1) \\
& \stackrel{d}{\to} &\left\VERT \mathcal G {\VERT T_2(C,\Delta)\VERT_2} - \frac{T_2(\mathcal G,\Delta)\widetilde{\otimes}T_1(C,T_2(C,\Delta))}{\VERT T_2(C,\Delta)\VERT_2} \right\VERT_2^2 \\
&& - {\VERT T_1(\mathcal G,T_2(C,\Delta)) - T_1(C,T_2(\mathcal G,\Delta)\VERT_2^2},
\end{eqnarray*}
and Theorem \ref{thm:null_dist} follows from Slutzky's Lemma
noting that 
$$
\widehat{D}_N = \frac{\phi(\widehat{C}_N) -\phi( C) }{\VERT T_2(\widehat{C}_N,\Delta)\VERT_2^2}
= \frac{\phi(\widehat{C}_N) }{\VERT T_2(\widehat{C}_N,\Delta)\VERT_2^2} .
$$
\end{proof}
}

 {
In the following let $q_{1-\alpha}$ be the $100\alpha \%$ quantile of the limiting random variable in Theorem \ref{thm:null_dist}, then an asymptotic level $\alpha$ test for the hypothesis in \eqref{H0A} is obtained by rejecting $H_0$, whenever 
\begin{equation}
   \label{testasy}
N\widehat{D}_N > q_{1-\alpha}.
\end{equation}
The next  result provides the asymptotic distribution under  the  alternative and  yields as a consequence the consistency of this test.}

\begin{thm}
\label{thm:asym_dist}
If ~$\E\| X \|_2^4 < \infty$, then the statistic $$
\sqrt{N}\left(\widehat{D}_N - D_0\right)
$$
converges in distribution to a centered normal distribution with variance 
{\begin{equation} \label{nusquare}
\nu^2 := 4 \left\langle \Gamma (A-B), (A-B)\right\rangle_{HS},
\end{equation}
where
\begin{align*}
A = &  C - \frac{T_2(C,\Delta) \widetilde{\otimes} T_1(C,T_2(C,\Delta))}{\VERT T_2(C,\Delta)\VERT_2^2},\\
B = &\frac{1}{\VERT T_2(C,\Delta)\VERT_2^2}\left[T_2(C,T_1(C,T_2(C,\Delta)))\widetilde{\otimes}\Delta - \frac{\VERT T_1(C,T_2(C,\Delta))\VERT_2^2}{\VERT T_2(C,\Delta)\VERT_2^2}T_2(C,\Delta)\widetilde{\otimes}\Delta\right],
\end{align*}}
and
the centering term $D_0$
 is defined in \eqref{mindistopt}.
\end{thm}
         
\begin{proof}
Observing \eqref{mindistopt}  and \eqref{estDN} we write
\begin{equation}
\label{eq:dnh0}   
\sqrt{N}\left(\widehat{D}_N - D_0\right)  
= \sqrt{N}\left( \VERT \widehat{C}_N \VERT_2^2 - \frac{\VERT T_1(\widehat{C}_N,T_2(\widehat{C}_N,\Delta))\VERT^2}{\VERT T_2(\widehat{C}_N,\Delta) \VERT_2^2} -  \VERT C \VERT_2^2 + \frac{\VERT T_1(C,T_2(C,\Delta))\VERT^2}{\VERT T_2(C,\Delta) \VERT_2^2}\right)
\end{equation}   
and  note that $\widehat{D}_N$ and $D_0$ are functions of the random variables
\begin{align}
\label{GN}
G_N = & \left(\VERT \widehat{C}_N \VERT_2^2 ,~{\VERT T_1(\widehat{C}_N,T_2(\widehat{C}_N,\Delta))\VERT^2}, ~{\VERT T_2(\widehat{C}_N,\Delta) \VERT_2^2}\right)^T,\\
\label{G}
G = & \left(\VERT C \VERT_2^2 ,~{\VERT T_1(C,T_2({C},\Delta))\VERT^2}, ~{\VERT T_2({C},\Delta) \VERT_2^2}\right)^T,
\end{align}
respectively. Therefore, we first investigate the weak convergence of the vector $\sqrt{N} (G_N - G)$. For this purpose we note that for $K,L \in S_2(H)$, the identity $$\VERT K\VERT_2^2 - \VERT L \VERT_2^2 = \VERT K-L \VERT_2^2 + 2\langle K-L,L \rangle_{HS} $$ holds and introduce the decomposition
$$\sqrt{N}(G_N - G) = I_N + II_N~, $$
where the random variables $I_N$ and $II_N$ are defined by
\begin{align*}
I_N = &\sqrt{N}\left(\VERT \widehat{C}_N - C\VERT_2^2 ,~{\VERT T_1(\widehat{C}_N,T_2(\widehat{C}_N,\Delta)) - T_1(C,T_2(C,\Delta))\VERT_2^2}, ~{\VERT T_2(\widehat{C}_N,\Delta) - T_2(C,\Delta)\VERT_2^2}\right)^T,\\
II_N = & 2\sqrt{N}\left(\left\langle \Cn - C,C\right\rangle_{HS}, ~ \left\langle T_1(\Cn,T_2(\Cn,\Delta))-T_1(C,T_2(C,\Delta)),T_1(C,T_2(C,\Delta))\right\rangle_{HS},\right.\\
&~~~~~~~~~~~ \left.\left\langle T_2(\Cn,\Delta) - T_2(C,\Delta), T_2(C,\Delta)\right\rangle_{HS}\right)^T.
\end{align*}
Using the linearity of $T_1$ and$T_2$ we further write
\begin{align*}
T_1(\widehat{C}_N,T_2(\widehat{C}_N,\Delta)) - T_1(C,T_2(C,\Delta)) =& T_1(\widehat{C}_N,T_2(\widehat{C}_N,\Delta)) - T_1(C,T_2(\widehat{C}_N,\Delta)) \\
+ & T_1(C,T_2(\Cn,\Delta))- T_1(C,T_2(C,\Delta))\\
= &T_1(\widehat{C}_N-C,T_2(\widehat{C}_N,\Delta)) + T_1(C,T_2(\Cn-C,\Delta))
\end{align*}
and  obtain the representation
\begin{align*}
I_N = &\frac{1}{\sqrt{N}}\left(\begin{array}{c}
\Big\VERT \sqrt{N}(\Cn -C) \Big\VERT_2^2 \\
\Big\VERT T_1\left(\sqrt{N}(\Cn - C),T_2(\Cn,\Delta)\right) + T_1\left(C,T_2(\sqrt{N}(\Cn-C),\Delta)\right) \Big\VERT_2^2\\
\Big\VERT T_2\left(\sqrt{N}(\Cn-C),\Delta\right)\Big\VERT_2^2
\end{array}\right) \\
=: &
\frac{1}{\sqrt{N}}F_1\left(\sqrt{N}(\Cn-C),\Cn\right),\\
II_N = &2\left(\begin{array}{c}
\left\langle \sqrt{N}(\Cn - C),C\right\rangle_{HS}\\
\left\langle T_1\left(\sqrt{N}(\Cn - C),T_2(\Cn,\Delta)\right) + T_1\left(C,T_2(\sqrt{N}(\Cn-C),\Delta)\right),T_1(C,T_2(C,\Delta))\right\rangle_{HS}\\
\left\langle T_2\left(\sqrt{N}(\Cn - C),\Delta\right),T_2(C,\Delta)\right\rangle_{HS}
\end{array}\right) \\
=: &
F_2\left(\sqrt{N}(\Cn-C),\Cn\right),
\end{align*}
where the last equations define the functions $F_1$ and $F_2$ in an obvious manner.
Note that 
$$
F := (F_1,F_2) : S_2(H) \times S_2(H) \mapsto \R^6
$$
is composition of continuous functions and hence continuous. By Proposition 5 in \citet{dauxois1982asymptotic}  the random variable $\sqrt{N}(\Cn - C)$ converges in distribution to a centered Gaussian random element $\mathcal{G}$  with variance 
\eqref{limop}
 in $S_2(H)$ with respect to Hilbert-Schmidt topology. Therefore, using continuous mapping arguments we have 
 $$
 F\left(\sqrt{N}(\Cn-C),\Cn\right) \stackrel{d}{\to} F(\mathcal{G},C)~,
 $$
  and consequently
$$\sqrt{N}(G_N - G) = \frac{1}{\sqrt{N}}F_1\left(\sqrt{N}(\Cn-C),\Cn\right) + F_2\left(\sqrt{N}(\Cn-C),\Cn\right) \stackrel{d}{\to} F_2(\mathcal{G},C).$$
We write
\begin{align*}
F_2(\mathcal{G},C) = 2\left(\begin{array}{c}
\left\langle \mathcal{G},C\right\rangle_{HS}\\
\left\langle T_1\left(\mathcal{G},T_2(C,\Delta)\right) + T_1\left(C,T_2(\mathcal{G},\Delta)\right),T_1(C,T_2(C,\Delta))\right\rangle_{HS}\\
\left\langle T_2\left(\mathcal{G},\Delta\right),T_2(C,\Delta)\right\rangle_{HS}
\end{array}\right)~,
\end{align*}
which can be further simplified as
\begin{align}
\label{eq:f2gc}
F_2(\mathcal{G},C) = 2\left(\begin{array}{c}
\left\langle \mathcal{G},C\right\rangle_{HS}\\
\left\langle \mathcal{G},  T_2(C,\Delta) \tilde{\otimes} T_1(C,T_2(C,\Delta))\right\rangle_{HS} + \left\langle C, T_2(\mathcal{G} ,\Delta)   \tilde{\otimes}  T_1(C,T_2(C,\Delta)) \right\rangle_{HS} \\
\langle \mathcal{G},T_2(C,\Delta) \tilde{\otimes} \Delta \rangle_{HS}
\end{array}\right)~.
\end{align}
By Proposition \ref{prop:t1} $T_2(\mathcal{G},\Delta)$ is a Gaussian Process in $S_2(H_2)$. This fact along with Lemma \ref{lem:gauss_ten} in Appendix \ref{sec5} imply that $F_2(\mathcal{G},C)$ is a normal distributed random vector with mean zero and covariance matrix, say $\Sigma$. 
By \eqref{eq:dnh0}, 
$$
\sqrt{N}\left(\widehat{D}_N -D_0\right) = \sqrt{N}(f(G_N) - f(G))~, 
$$
 where the function $f: \R^3 \mapsto \R$ is defined by $f(x,y,z) = x - y/z$ and $G_N$ and $G$ are defined in \eqref{GN} and \eqref{G}, respectively.
 Therefore, using the delta method and the fact that $$\P\left( \VERT T_2(\Cn,\Delta) \VERT_2^2 > 0\right) \to \P\left( \VERT T_2(C,\Delta) \VERT_2^2 > 0\right) = 1$$ as $\VERT C \VERT_2 \neq 0$, we finally obtain
\begin{equation}
\sqrt{N}\left(\widehat{D}_N - D_0 \right)\stackrel{d}{\to} N\left(0,(\nabla f(G))^T \Sigma (\nabla f(G))\right)
\end{equation}
as $n \to \infty$ where $\nabla f(x,y,z) = (1, - 1/z,      y/z^2)^T$ denotes the gradient of the function $f$.
{
Finally,  the proof of the representation \eqref{nusquare} of the limiting variance is given in Section \ref{a6} in the Appendix.}
\end{proof}

{
\begin{Remark} {\rm 
If the null hypothesis is true, i.e., $C = C_1 \otimes C_2$, the variance $\nu^2$ in Theorem \ref{thm:asym_dist} becomes zero. 
Indeed, under the null hypothesis of separability we have
\begin{align*}
    T_2(C,\Delta) \widetilde{\otimes} T_1(C,T_2(C,\Delta)) = & 
      \langle C_2,\Delta\rangle^2 C_1 \widetilde{\otimes} T_1(C,C_1) =\langle C_2,\Delta\rangle^2 \langle C_1, C_1 \rangle  C_1 \widetilde{\otimes} C_2
    \\
    = &   \VERT \langle C_2,\Delta\rangle C_1 \VERT_2^2 C
    =  \VERT T_2(C,\Delta)\VERT_2^2 C,
\end{align*}
which implies  $A = 0$ for the quantity $A$ in Theorem \ref{thm:asym_dist}. Similarly, 
\begin{align*}
   & \frac{\VERT T_1(C,T_2(C,\Delta))\VERT_2^2}{\VERT T_2(C,\Delta)\VERT_2^2}T_2(C,\Delta) =  \frac{\VERT T_1(C,C_1)\VERT_2^2\langle C_2,\Delta \rangle^2}{\VERT T_2(C,\Delta)\VERT_2^2}T_2(C,\Delta)\\
 &~~~ ~~~~~~ ~~~  = \frac{\VERT \langle C_1,C_1 \rangle C_2\VERT_2^2\langle C_2,\Delta \rangle^2}{\VERT \langle C_2,\Delta \rangle C_1\VERT_2^2}T_2(C,\Delta)   =  \frac{\VERT  C_2\VERT_2^2\VERT  C_1\VERT_2^4\langle C_2,\Delta \rangle^2}{\VERT  C_1\VERT_2^2\langle C_2,\Delta \rangle^2}T_2(C,\Delta)\\
&~~~ ~~~~~~ ~~~     =  \VERT  C_2\VERT_2^2 \VERT  C_1\VERT_2^2T_2(C,\Delta)  =  \langle C_2,C_2 \rangle \langle C_1,C_1 \rangle \langle C_2,\Delta \rangle C_1   
\\
 &~~~ ~~~~~~ ~~~    =  \langle C_1,C_1 \rangle \langle C_2,\Delta \rangle T_2(C,C_2)  =  T_2(C, C_2\langle C_1,C_1 \rangle)\langle C_2,\Delta \rangle   =  T_2(C, T_1(C,C_1))\langle C_2,\Delta \rangle\\
 &~~~ ~~~~~~ ~~~    =  T_2(C, T_1(C,\langle C_2,\Delta \rangle C_1))
 =  T_2(C,T_1(C,T_2(C,\Delta)))
\end{align*}
and consequently the quantity  $B$ in Theorem \ref{thm:asym_dist} also vanishes.
Therefore under the null hypothesis $\sqrt{N}\widehat{D}_N \stackrel{p}{\to} 0$ (which is also a consequence of Theorem \ref{thm:null_dist}). 
}
\end{Remark}
}

 {
\begin{Remark} \label{rem3}{\rm
A sufficient condition for Theorem \ref{thm:null_dist} and 
\ref{thm:asym_dist} to hold is $\E\|X\|_2^4 < \infty $. {As indicated in Remark 2.2 (1) of  \citet{Aston2017}, this is a weaker than Condition 2.1 in their paper, which assumes
\begin{equation}
\label{aston}
\sum_{j=1}^{\infty}\left(\E\left[\langle X,e_j\rangle^4 \right]\right)^{1/4} < \infty,
\end{equation}
for some orthonormal basis $(e_j)_{j\geq 1}$ of $H$.
}
 These authors used this assumption 
to prove  weak convergence  under  the  trace-norm topology, which is required to establish Theorem 2.3 in \citet{Aston2017}. In contrast the proof of Theorem \ref{thm:asym_dist} here only requires weak 
convergence  under the Hilbert-Schmidt topology, which defines a weaker topology.}
\end{Remark}
}

{
\begin{Remark}
{\rm 
Note that the asymptotic distribution depends (under the null hypothesis and alternative) on the operator $\Delta$.
Under the  assumptions of Theorem \ref{thm:asym_dist}, we 
obtain an approximation of the power of the test \eqref{testasy} by 
\begin{align*}
\P\left({N}\widehat{D}_N > q_{1-\alpha}\right) = &\P\left(\sqrt{N}(\widehat{D}_N - D_0 ) > \frac{q_{1-\alpha}}{\sqrt{N}}-\sqrt{N}D_0 \right)\\ 
\approx  &1 - \Phi\left(\frac{q_{1-\alpha}}{\sqrt{N}\nu} - \frac{\sqrt{N}D_0}{\nu}\right),
\end{align*}
where $\Phi$ is the standard normal distribution function
and $\nu^2$ is defined by \eqref{nusquare}. Under the alternative  
$D_0$ is  positive. 
Therefore the rejection probability converges to $1$ with increasing  sample size $N$ and consequently the proposed test is consistent. \\
Moreover, if $N$ is sufficiently large, the power is a decreasing function of  the variance $\nu^2$ in \eqref{nusquare}. As this quantity depends on the operator $\Delta$ it is desirable to choose $\Delta$
such that  $v^2$ is minimal. The   solution of this optimization problem  depends on the unknown covariance operator $C$ and it seems to be intractable to obtain it explicitly. 
However, we will demonstrate in Section \ref{sec4} that 
for finite sample sizes the resulting tests are not very sensitive with respect to the choice of the operator $\Delta$.
}
\end{Remark}}
\medskip

\subsection{Hilbert-Schmidt integral operators}
\label{sec32}
In the remaining part of this section we concentrate on the case, where  $X$ is a random element in  $H=L^2(S \times T)$ and {$S \subset \R^p$ and $T \subset \R^q$ are bounded measurable sets.} In this particular scenario, we choose $\Delta$ also to be an integral operator generated by a kernel $\psi(t,t^{\prime})$. With this choice, using the explicit formula for the operator $T_1$ described in Proposition \ref{prop:t1kernel} the minimum distance can be expressed in terms of the corresponding kernels, that is
\begin{eqnarray}
\label{eq:dker}
D_0 =  D(T_2(C,\Delta)) &=&
\int_T\int_T\int_S\int_S c^2(s,t,s',t')dsds'dtdt'  \\
\nonumber 
&-& \frac{\int_T\int_T\left[\int_S\int_S c(s,t,s't')\tilde c_1(s,s')dsds'\right]^2dtdt'}{\int_S\int_S \tilde c_1^2(s,s')dsds'}.
\end{eqnarray}
where $ \tilde c_1$ denotes the kernel corresponding to the operator
$T_2(C,\Delta)$, that is
$$
\tilde c_1  (s,s^\prime ) = \int_T \int_T c(s,t,s^\prime,t^{\prime} )\psi(t,t^{\prime}) dt dt^{\prime}~.
$$
In this case the estimator $\Cn$ defined in \eqref{estgen}
is induced by the kernel
$$\hat{c}_N(s,t,s',t') = \frac{1}{N}\sum_{i=1}^N(X_i(s,t)-\overline{X}(s,t))(X_i(s',t')-\overline{X}(s',t'))~,$$
and the estimator $\widehat{C}_{1N} = T_2(\widehat{C}_N,\Delta)$
defined in  \eqref{estgenc1}
is induced by the kernel
$$\hat{c}_{1N} (s,s^\prime )  = \int_T \int_T \hat{c}_N(s,t,s',t')\psi(t,t')dt dt'.
$$
The estimator $\widehat D_N$ of $D_0$ is calculated by plugging in $\hat{c}_N$ and $\hat{c}_{1N}$ to the expression in \eqref{eq:dker}.

{
\begin{Remark}
\label{rem4}
{\rm ~
\begin{itemize}
    \item [(a)]
A natural choice for  $\Delta$ is an operator with  constant kernel, that is
$\psi (t,t^\prime ) \equiv 1$, which gives for the kernel of the operator $T_1(C,\Delta)$ 
$$
 \int_T \int_T c(s,t,s',t')dtdt',
$$
This operator has to be distinguished form  partial trace, which is defined as the integral operator with kernel
$$
 \int_T c(s,t,s',t)dt.
$$
and was used by \citet{Aston2017}.
\item[(b)]
Although the proposed  estimator is based on the norm of the complete covariance kernel $c$, numerically we do not need to store the complete covariance kernel. 
 For example, we obtain for the first term of the statistic
 $\widehat D_N$ the representation
\begin{align*}
\VERT \Cn \VERT_2^2 = &\frac{1}{N^2}\int_T\int_S\int_T\int_S\left(\sum_{i=1}^N (X_i(s,t)-\overline{X}(s,t))(X_i(s',t') - \overline{X}(s',t')) \right)^2 ds dt ds' dt'\\
= &\frac{1}{N^2}\sum_{i=1}^N \sum_{j=1}^N \left[\int_T\int_S (X_i(s,t)-\overline{X}(s,t))(X_j(s,t) - \overline{X}(s,t))ds dt\right]^2
\end{align*}
All other  terms of the estimator in \eqref{estDN}  can be  represented similarly  using simple matrix operations on the data matrix without storing the full or marginal covariance kernels.
\end{itemize}
}
\end{Remark}}

{
\subsection{Bootstrap test for separability}
\label{sec:test}
An obvious method for testing the hypothesis of separability is based on the quantiles of the limiting random variable given in Theorem \ref{thm:null_dist}. For this purpose one can  estimate the limiting covariance operator $\Gamma$ from the data and simulate a centered Gaussian process $\mathcal{G}$ with covariance operator $\Gamma$. The limiting distribution can then be calculated as function of the simulated Gaussian processes. The simulated $100(1-\alpha)\%$ quantile is then compared to $N\widehat{D}_N$ to test the null hypothesis, which gives the test \eqref{testasy}. It turns out that this approach provides a very powerful test for the hypothesis of separability (see the empirical results in Section \ref{sec4}).

As this method requires the estimation of the covariance kernel $\Gamma$ we also propose a bootstrap test. The simplest method would be to approximate the limiting distribution of $N\widehat{D}_N$ by
the distribution of  the statistic 
$\{N\hat D_N^* - N\widehat{D}_N\}$, where $\hat D_N^*$ is the test statistic calculated from  a bootstrap sample drawn from $X_1, \ldots , X_N$ with replacement.  

However, this procedure fails to give good power under the alternative. This observation can be explained by studying the test statistic a little more closely. In general, we can write
\begin{align*}
    \widehat{D}_N - D_0 =& \VERT \widehat{C}_N \VERT_2^2 - \frac{\VERT T_1(\widehat{C}_N,T_2(\widehat{C}_N,\Delta))\VERT_2^2}{\VERT T_2(\widehat{C}_N,\Delta) \VERT_2^2} - \VERT C \VERT_2^2 + \frac{\VERT T_1(C,T_2(C,\Delta))\VERT_2^2}{\VERT T_2(C,\Delta) \VERT_2^2}\\
    = &A_{1,N} + A_{2,N}
\end{align*}
where the statistics $A_{1,N}$ and $  A_{2,N}$ are given by 
\begin{align}
\label{eq:decomp}
    A_{1,N} = &~\VERT \widehat{C}_N - C \VERT_2^2 - \frac{\VERT T_1(\widehat{C}_N - C,T_2(\widehat{C}_N,\Delta)) \VERT_2^2 + \VERT T_1(C,T_2(\widehat{C}_N - C,\Delta)) \VERT_2^2}{\VERT T_2(\widehat{C}_N,\Delta) \VERT_2^2}\\
    &~+ \frac{\VERT T_1(C,(T_2(C,\Delta)))\VERT_2^2}{\VERT T_2(\widehat{C}_N,\Delta) \VERT_2^2 \VERT T_2(C,\Delta) \VERT_2^2} \VERT T_2(\widehat{C}_N - C,\Delta) \VERT_2^2 , \nonumber\\
    \label{eq:decomp1}
     A_{2,N} = &~2\langle \widehat{C}_N - C,C \rangle_{HS} - \frac{2\langle T_1(\widehat{C}_N - C,T_2(\widehat{C}_N,\Delta)),  T_1( C,T_2(\widehat{C}_N,\Delta))\rangle}{\VERT T_2(\widehat{C}_N,\Delta) \VERT_2^2}\\
    &~- \frac{2\langle T_1( C,T_2(\widehat{C}_N-C,\Delta)),  T_1( C,T_2(C,\Delta))\rangle_{S_2(H_2)}}{\VERT T_2(\widehat{C}_N,\Delta) \VERT_2^2}\nonumber\\
    & ~+\frac{2\VERT T_1(C,(T_2(C,\Delta)))\VERT_2^2}{\VERT T_2(\widehat{C}_N,\Delta) \VERT_2^2 \VERT T_2(C,\Delta) \VERT_2^2} \langle T_2(\widehat{C}_N - C,\Delta),T_2(C,\Delta) \rangle_{S_2(H_1)},
    \nonumber
\end{align}
respectively.
If the true underlying covariance operator $C$ is separable, then $  A_{2,N} =  0$ and hence only the first term contributes to the limiting null distribution. Now note that  a similar decomposition  for the bootstrap statistic gives 
$${D}_N^* - \widehat{D}_N =   A_{1,N}^*  + A_{2,N}^*,$$
where $A_{1,N}^*$ and $A_{2,N}^*$ are defined similarly as in \eqref{eq:decomp} and \eqref{eq:decomp1} replacing $\widehat{C}_N$ by its  bootstrap analogue $\widehat{C}_N^*$ and $C$ by $\widehat{C}_N$, respectively. The first term $NA_{1,N}^* $ can be shown to approximate the limiting distribution of $NA_{1,N}$, which is the desired null limiting distribution. However, the estimate $\widehat{C}_N$ is in general not separable. As consequence
$N A_{2,N}^*$ is not zero and   a simple bootstrap 
using the quantile of the distribution of
$N(\hat D_N^* - \widehat{D}_N)$ will result in a test with very low power. }

 {To avoid this problem, instead of using the quantile of 
$N(\hat D_N^* - \widehat{D}_N)$  we  propose to use the quantile of the distribution 
of $NA_{1,N}^{*}$. This quantile  can be estimated by the 
empirical quantile from the  bootstrap sample  $NA_{1,N}^{1*}, \ldots , NA_{1,N}^{B*}$ 
(here$ NA_{1,N}^{b*}$ is the corresponding statistic calculated from the $b$th bootstrap sample, for $b=1, \ldots , B$).} 

{\color{red}
\begin{table}[t]
\begin{center}
\caption{\it Empirical rejection probabilities of different tests for the hypothesis of separability (level $5\%$). 
M1: the bootstrap test proposed in this paper, 
M2: the asymptotic test \eqref{testasy} proposed  in this paper, {(the different kernels are indicated in brackets for both M1 and M2)};
The data are generated from Gaussian distribution with zero mean and covariance kernel given by \eqref{eq:ker_Gneitning}, where the case $\beta=0$ corresponds to the null hypothesis of separability.}
\label{Gneitning_Gauss}
\bigskip
\begin{tabular}{|c|c|c|c|c|c|c|c|}
\hline
$\beta$ & N &  M1 ($\psi_1$)  & M1 ($\psi_2$) & M1 ($\psi_3$)  & M2 ($\psi_1$)& M2($\psi_2$) & M2($\psi_3$) \\
\hline
0 & 100 & 3.7 \% & 4.5\% & 3.9\% & 4.6\% & 4.3\% & 4.5\%  \\
0 & 200 & 3.5\% & 4.0\% & 4.2\% & 4.9\% & 4.6\% & 4.8\% \\
0 & 500 & 4.1\% & 5.5\% & 3.6\% & 4.2\% & 4.8\% & 4.5\% \\
\hline
0.3 & 100 & 30.3 \% & 34.6\% & 31.7\% & 53.8\% & 48.1\% & 49.3\%  \\
0.3 & 200 & 45.4 \% & 41.2\% & 47.9\% & 71.6\% & 73.4\% & 72.5\% \\
0.3 & 500 & 52.7\% &53.2\% & 50.5\% & 81.8\% & 77.4\% & 79.3\% \\
\hline
0.5 & 100 & 46.9 \% & 47.8\% & 50.2\% & 63.4\% & 65.1\% & 64.8\% \\
0.5 & 200 & 70.3\% & 70.3\% & 71.0\% & 87.9\% & 88.1\% & 87.1\%  \\
0.5 & 500 & 88.9 \% & 83.2\% & 89.1\% & 93.5\% & 91.6\% & 94.2\% \\
\hline
0.7 & 100 & 66.2 \% & 68.3\% & 67.8\% & 79.4\% & 80.2\% & 81.6\% \\
0.7 & 200 & 82.8 \% & 79.7 \% & 83.5\% & 93.4 \% & 95.7\% & 92.9\% \\
0.7 & 500 &  89.1 \% & 92.5\% & 93.3 \% & 100\% & 100\% & 99.9\% \\
\hline
1 & 100 & 88.6\% & 90.4\% & 89.9\% & 97.3\% & 96.8\% & 96.2\%  \\
1 & 200 & 93.5\% & 94.6\% & 93.9 \%  & 100\% & 100\% & 100\% \\
1 & 500 & 96.8\% & 95.9\% & 97.1\% & 100\% & 100\% & 100\% \\
\hline
\end{tabular}
\end{center}
\end{table}

\begin{table}[t]
\begin{center}
\caption{
\label{tab2}
\it Empirical rejection probabilities of different tests for the hypothesis of separability (level $5\%$) based on \citet{Aston2017} (M3) and \citet{conkokrei2017}(M4); The number of spatial and temporal components used are indicated in the bracket. The data are generated from Gaussian distribution with zero mean and covariance kernel given by \eqref{eq:ker_Gneitning}, where the case $\beta=0$ corresponds to the null hypothesis of separability.}
\label{Gneitning_Gauss_Other}
\bigskip
\begin{tabular}{|c|c|c|c|c|c|c|c|}
\hline
$\beta$ & N &  M3 (k=2)  & M3 (k=3) & M3 (k=4)  & M4 (L=J=2)& M4 (L=J=3) & M4 (L=J=4) \\
\hline
0 & 100 & 4.1 \% & 4.3\% & 4.9\% & 4.2\% & 4.8\% & 5.5\%  \\
0 & 200 & 4.5\% & 4.2\% & 4.8\% & 5.1\% & 5.1\% & 5.3\% \\
0 & 500 & 5.0\% & 4.7\% & 5.1\% & 4.8\% & 5.3\% & 5.6\% \\
\hline
0.3 & 100 & 46.6 \% & 48.1\% & 51.2\% & 13.4\% & 32.2\% & 48.4\%  \\
0.3 & 200 & 62.0 \% & 71.5\% & 70.8\% & 22.8\% & 42.6\% & 59.7\% \\
0.3 & 500 & 65.8\% & 73.4\% & 81.2\% & 33.7\% & 53.3\% & 69.5\% \\
\hline
0.5 & 100 & 57.1 \% & 62.3\% & 70.1\% & 25.9\% & 42.6\% & 58.3\% \\
0.5 & 200 & 86.2\% & 88.4\% & 89.4\% & 43.4\% & 57.9\% & 79.8\%  \\
0.5 & 500 & 90.2 \% & 91.5\% & 93.2\% & 55.7\% & 71.2\% & 90.5\% \\
\hline
0.7 & 100 & 73.6 \% & 76.1\% & 80.1\% & 48.4\% & 61.8\% & 75.4\% \\
0.7 & 200 & 90.1 \% & 92.4 \% & 95.8\% & 59.3 \% & 76.3\% & 93.9\% \\
0.7 & 500 &  99.7 \% & 99.9\% & 100 \% & 66.5\% & 91.4\% & 100\% \\
\hline
1 & 100 & 92.7\% & 96.7\% & 100\% & 59.7\% & 95.6\% & 100\%  \\
1 & 200 & 99.4\% & 100\% & 100 \%  & 80.1\% & 99.9\% & 100\% \\
1 & 500 & 100\% & 100\% & 100\% & 99.9\% & 100\% & 100\% \\
\hline
\end{tabular}
\end{center}
\end{table}

\begin{table}[t]
\begin{center}
\caption{
\it Empirical rejection probabilities of different tests for the hypothesis of separability (level $5\%$).
M1: the bootstrap test proposed in this paper, 
M2: the asymptotic test \eqref{testasy} proposed  in this paper {(the different kernels are indicated in brackets for both M1 and M2)}. The data are generated from t distribution with 5 degrees of freedom and covariance kernel given by \eqref{eq:ker_Gneitning}, where the case $\beta=0$ corresponds to the null hypothesis of separability.}
\label{Gneitning_t}
\bigskip
\begin{tabular}{|c|c|c|c|c|c|c|c|}
\hline
$\beta$ & N &  M1 ($\psi_1$)  & M1 ($\psi_2$) & M1 ($\psi_3$)  & M2 ($\psi_1$)& M2($\psi_2$) & M2($\psi_3$)  \\
\hline
0 & 100 & 5.3 \% & 5.8\% & 3.8\% & 6.1\% & 5.5\% & 4.3\%  \\
0 & 200 & 4.7\% & 3.4\% & 3.2\% & 4.9\% & 3.7\% & 3.4\%  \\
0 & 500 & 4.2\% & 3.9\% & 5.3\% & 5.2\% & 5.7\% & 4.9\%  \\
\hline
0.3 & 100 & 13.4 \% & 15.8\% & 12.9\% & 31.3\% & 29.6\% & 28.4\% \\
0.3 & 200 & 23.7 \% & 26.2\% & 25.1\% & 35.3\% & 37.2\% & 36.9\% \\
0.3 & 500 & 47.7\% &48.3\% & 50.9\% & 65.4\% & 67.1\% & 64.9\% \\
\hline
0.5 & 100 & 29.8 \% & 35.1\% & 30.2\% & 44.8\% & 48.1\% & 41.6\%  \\
0.5 & 200 & 49.7\% & 50.8\% & 51.4\% & 66.2\% & 58.4\% & 63.9\%  \\
0.5 & 500 & 86.4 \% & 82.9\% & 88.1\% & 92.5\% & 90.8\% & 91.1\% \\
\hline
0.7 & 100 & 48.3 \% & 46.5\% & 47.1\% & 64.2\% & 59.8\% & 61.7\% \\
0.7 & 200 & 69.4 \% & 70.1\% & 68.9\% & 86.8 \% & 88.4\% & 89.7\% \\
0.7 & 500 &  92.4 \% & 93.6\% & 91.8\% & 95.7\% & 94.9\% & 96.2\% \\
\hline
1 & 100 & 72.1\% & 68.4\% & 67.9\% & 83.5\% & 85.8\% & 89.6\%  \\
1 & 200 & 90.1\% & 91.4\% & 88.6 \%  & 95.9\% & 97.1\% & 95.4\% \\
1 & 500 & 96.4\% & 95.3\% & 95.1\% & 99.8\% & 100\% & 99.9\% \\
\hline
\end{tabular}
\end{center}
\end{table}

\begin{table}[t]
\begin{center}
\caption{\label{tab4}
\it Empirical rejection probabilities of different tests for the hypothesis of separability (level $5\%$) based on \citet{Aston2017} (M3) and \citet{conkokrei2017}(M4); The number of spatial and temporal components used are indicated in the bracket. The data are generated from $t$ distribution with 5 degrees of freedom and covariance kernel given by \eqref{eq:ker_Gneitning}, where the case $\beta=0$ corresponds to the null hypothesis of separability.}
\label{Gneitning_t_Other}
\bigskip
\begin{tabular}{|c|c|c|c|c|c|c|c|}
\hline
$\beta$ & N &  M3 (k=2)  & M3 (k=3) & M3 (k=4)  & M4 (L=J=2)& M4 (L=J=3) & M4 (L=J=4) \\
\hline
0 & 100 & 6.4 \% & 5.1\% & 5.6\% & 3.5\% & 4.3\% & 5.0\%  \\
0 & 200 & 3.2\% & 5.3\% & 4.8\% & 4.8\% & 5.2\% & 5.6\% \\
0 & 500 & 4.1\% & 4.7\% & 4.9\% & 4.8\% & 5.4\% & 5.5\% \\
\hline
0.3 & 100 & 15.7 \% & 21.4\% & 29.3\% & 6.8\% & 15.7\% & 24.9\%  \\
0.3 & 200 & 24.3 \% & 32.9\% & 41.6\% & 11.0\% & 23.2\% & 38.8\% \\
0.3 & 500 & 51.8\% &58.7\% & 62.5\% & 26.7\% & 44.6\% & 58.7\% \\
\hline
0.5 & 100 & 35.2 \% & 38.6\% & 41.3\% & 17.9\% & 25.4\% & 41.6\% \\
0.5 & 200 & 52.5\% & 68.4\% & 73.7\% & 31.1\% & 45.2\% & 61.9\%  \\
0.5 & 500 & 87.6 \% & 92.6\% & 96.4\% & 49.2\% & 60.7\% & 79.3\% \\
\hline
0.7 & 100 & 48.9 \% & 52.3\% & 59.9\% & 27.6\% & 49.9\% & 62.6\% \\
0.7 & 200 & 73.4 \% & 80.2 \% & 84.6\% & 38.4 \% & 53.8\% & 84.7\% \\
0.7 & 500 &  93.8 \% & 95.3\% & 98.7 \% & 56.2\% & 70.9\% & 90.1\% \\
\hline
1 & 100 & 71.7\% & 73.4\% & 82.4\% & 59.7\% & 71.3\% & 89.6\%  \\
1 & 200 & 92.8\% & 96.1\% & 99.9 \%  & 76.5\% & 88.4\% & 95.9\% \\
1 & 500 & 98.7\% & 100\% & 100\% & 87.6\% & 98.9\% & 100\% \\
\hline
\end{tabular}
\end{center}
\end{table}


\begin{table}[t]
\begin{center}
\caption{\it Empirical rejection probabilities of different tests for the hypothesis of separability (level $5\%$).
M1: the bootstrap test proposed in this paper, 
M2: the asymptotic test \eqref{testasy} proposed  in this paper {(the different kernels are indicated in brackets for both M1 and M2)}. The data are generated from Gaussian distribution with zero mean and covariance kernel given in \eqref{eq:ker_ch}, where the case $c_0=1$ corresponds to the null hypothesis of separability.}
\label{CH_Gauss}
\bigskip
\begin{tabular}{|c|c|c|c|c|c|c|c|}
\hline
$c_0$ & N &  M1 ($\psi_1$)  & M1 ($\psi_2$) & M1 ($\psi_3$)  & M2 ($\psi_1$)& M2($\psi_2$) & M2($\psi_3$)  \\
\hline
1 & 100 & 4.1 \% & 4.3\% & 3.9\% & 4.8\% & 3.8\% & 4.7\% \\
1 & 200 & 3.6\% & 3.8\% & 4.1\% & 4.8\% & 5.2\%  & 5.4\% \\
1 & 500 & 3.3\% & 3.6\% & 5.7\% & 5.3\% & 4.2\% & 4.8\%  \\
\hline
3 & 100 & 28.9 \% & 33.2\% & 28.1\% & 49.4 \% & 47.1\% & 48.5\% \\
3 & 200 & 39.5 \% & 40.6\% & 42.1\% & 68.5\% & 67.2\% & 69.4\% \\
3 & 500 & 57.2\% & 52.8\% & 55.4\% & 77.1\% & 73.9\% & 75.8\% \\
\hline
5 & 100 & 50.7 \% & 51.3\% & 54.8\% & 71.8\% & 70.1\% & 76.2\%  \\
5 & 200 & 76.1 \% & 72.4\% & 74.7\% & 90.2\% & 89.6\% & 91.3\%  \\
5 & 500 & 83.8\% & 84.9\% & 83.4\% & 99.1\% & 98.7\% & 99.4\% \\
\hline
7 & 100 & 68.4 \% & 68.9\% & 67.5 \% & 83.4\% & 82.7\% & 85.1\% \\
7 & 200 & 87.2\% & 89.1\% & 81.3\% & 94.2\% & 95.3\% & 94.4\%\\
7 & 500 &  97.3 \% & 94.6\% & 95.2\% & 99.9\% & 100\% & 100\% \\
\hline
10 & 100 & 83.8\% & 85.2\% & 84.7\% & 97.4\% & 96.9\% & 98.1\%   \\
10 & 200 & 89.2 \% & 92.6\% & 91.8\% & 100\% & 100\% & 100\% \\
10 & 500 & 97.7\% & 95.8\% & 98.3\% & 100\% & 100\% & 100\% \\
\hline
\end{tabular}
\end{center}
\end{table}

\begin{table}[t]
\begin{center}
\caption{\label{tab6}
\it Empirical rejection probabilities of different tests for the hypothesis of separability (level $5\%$) based on \citet{Aston2017} (M3) and \citet{conkokrei2017}(M4); The number of spatial and temporal components used are indicated in the bracket. The data are generated from Gaussian distribution with zero mean and covariance kernel given by \eqref{eq:ker_ch}, where the case $c_0=1$ corresponds to the null hypothesis of separability.}
\label{CH_Gauss_Other}
\bigskip
\begin{tabular}{|c|c|c|c|c|c|c|c|}
\hline
$c_0$ & N &  M3 (k=2)  & M3 (k=3) & M3 (k=4)  & M4 (L=J=2)& M4 (L=J=3) & M4 (L=J=4) \\
\hline
1 & 100 & 5.3 \% & 5.2\% & 5.8\% & 7.7\% & 6.5\% & 7.4\%  \\
1 & 200 & 4.5\% & 4.9\% & 5.2\% & 6.9\% & 7.1\% & 7.5\% \\
1 & 500 & 3.9\% & 4.8\% & 4.4\% & 6.5\% & 6.8\% & 7.2\% \\
\hline
3 & 100 & 38.7 \% & 43.8\% & 48.5\% & 9.8\% & 14.3\% & 33.1\%  \\
3 & 200 & 60.9 \% & 68.3\% & 72.1\% & 15.2\% & 24.6\% & 41.2\% \\
3 & 500 & 65.9\% & 70.4\% & 73.3\% & 26.4\% & 33.4\% & 52.7\% \\
\hline
5 & 100 & 62.5 \% & 64.1\% & 68.2\% & 39.8\% & 51.5\% & 63.9\% \\
5 & 200 & 85.4\% & 89.7\% & 91.4\% & 56.3\% & 62.4\% & 77.1\%  \\
5 & 500 & 95.1 \% & 96.3\% & 98.5\% & 78.2\% & 83.6\% & 94.8\% \\
\hline
7 & 100 & 77.3 \% & 80.2\% & 81.1\% & 64.9\% & 73.2\% & 84.9\% \\
7 & 200 & 92.8 \% & 94.1 \% & 96.3\% & 85.2 \% & 92.3\% & 98.7\% \\
7 & 500 &  99.9 \% & 100\% & 100\% & 93.6\% & 97.6\% & 100\% \\
\hline
10 & 100 & 94.2\% & 96.8\% & 99.9\% & 77.1\% & 86.9\% & 95.8\%  \\
10 & 200 & 100\% & 100\% & 100 \%  & 90.8\% & 100\% & 100\% \\
10 & 500 & 100\% & 100\% & 100\% & 97.1\% & 100\% & 100\% \\
\hline
\end{tabular}
\end{center}
\end{table}

\begin{table}[t]
\begin{center}
\caption{\it Empirical rejection probabilities of different tests for the hypothesis of separability (level $5\%$).
M1: the bootstrap test proposed in this paper, 
M2: the asymptotic test \eqref{testasy} proposed  in this paper {(the different kernels are indicated in brackets for both M1 and M2)};
M3: the test of \citet{Aston2017};  M4:  the test of \citet{conkokrei2017}. The data are generated from t-distribution with 5 degrees of freedom and covariance kernel given in \eqref{eq:ker_ch}, where the case $c_0=1$ corresponds to the null hypothesis of separability.}
\label{CH_t}
\bigskip
\begin{tabular}{|c|c|c|c|c|c|c|c|}
\hline
$c_0$ & N &  M1 ($\psi_1$)  & M1 ($\psi_2$) & M1 ($\psi_3$)  & M2 ($\psi_1$)& M2($\psi_2$) & M2($\psi_3$)  \\
\hline
1 & 100 & 2.3 \% & 2.4\% & 1.8\% & 3.3\% & 2.7\% & 3.6\% \\
1 & 200 & 2.7\% & 1.9\% & 3.4\% & 2.9\% & 3.6\%  & 4.8\%  \\
1 & 500 & 2.9\% & 3.8\% & 3.1\% & 4.3\% & 3.1\% & 3.6\% \\
\hline
3 & 100 & 3.9 \% & 4.6\% & 4.1\% & 6.3 \% & 5.8\% & 4.4\% \\
3 & 200 & 18.2 \% & 13.4\% & 12.9\% & 30.1\% & 27.5\% & 28.9\%\\
3 & 500 & 53.4\% & 52.5\% & 54.3\% & 68.9\% & 67.8\% & 71.2\%\\
\hline
5 & 100 & 8.5 \% & 7.9\% & 9.4\% & 13.3\% & 20.5\% & 17.4\%  \\
5 & 200 & 41.9 \% & 42.6\% & 46.2\% & 61.1\% & 58.3\% & 57.6\%  \\
5 & 500 & 82.9\% & 83.8\% & 85.6\% & 91.0\% & 94.4\% & 92.3\% \\
\hline
7 & 100 & 14.8 \% & 12.9\% & 13.5 \% & 23.8\% & 26.4\% & 25.9\% \\
7 & 200 & 59.3\% & 62.4\% & 63.3\% & 77.5\% & 76.1\% & 73.9\% \\
7 & 500 &  93.4 \% & 94.1\% & 92.8\% & 99.8\% & 100\% & 100\% \\
\hline
10 & 100 & 34.5\% & 33.6\% & 29.8\% & 47.2\% & 46.8\% & 48.3\%   \\
10 & 200 & 59.6 \% & 62.5\% & 63.2\% & 81.4\% & 80.9\% & 86.5\% \\
10 & 500 & 92.9\% & 96.2\% & 95.8\% & 100\% & 100\% & 100\% \\
\hline
\end{tabular}
\end{center}
\end{table}

\begin{table}[t]
\begin{center}
\caption{\label{tab8}
\it Empirical rejection probabilities of different tests for the hypothesis of separability (level $5\%$) based on \citet{Aston2017} (M3) and \citet{conkokrei2017}(M4); The number of spatial and temporal components used are indicated in the bracket. The data are generated from $t$ distribution with 5 degrees of freedom and covariance kernel given by \eqref{eq:ker_ch}, where the case $c_0=1$ corresponds to the null hypothesis of separability.}
\label{CH_t_Other}
\bigskip
\begin{tabular}{|c|c|c|c|c|c|c|c|}
\hline
$\beta$ & N &  M3 (k=2)  & M3 (k=3) & M3 (k=4)  & M4 (L=J=2)& M4 (L=J=3) & M4 (L=J=4) \\
\hline
1 & 100 & 2.9 \% & 3.6\% & 4.2\% & 1.7\% & 3.2\% & 4.9\%  \\
1 & 200 & 1.4\% & 3.8\% & 4.7\% & 3.1\% & 4.6\% & 5.8\% \\
1 & 500 & 4.2\% & 4.5\% & 4.6\% & 3.4\% & 4.4\% & 5.2\% \\
\hline
3 & 100 & 4.3 \% & 10.1\% & 15.4\% & 3.8\% & 9.6\% & 18.3\%  \\
3 & 200 & 15.8 \% & 20.8\% & 23.4\% & 7.3\% & 15.4\% & 28.7\% \\
3 & 500 & 64.6\% & 71.1\% & 76.5\% & 29.3\% & 34.6\% & 49.8\% \\
\hline
5 & 100 & 10.6 \% & 19.6\% & 22.8\% & 7.7\% & 15.3\% & 26.4\% \\
5 & 200 & 43.5\% & 49.9\% & 61.2\% & 32.8\% & 41.3\% & 52.9\%  \\
5 & 500 & 86.7 \% & 88.1\% & 89.8\% & 54.3\% & 63.2\% & 71.6\% \\
\hline
7 & 100 & 14.6 \% & 23.2\% & 28.5\% & 9.8\% & 18.2\% & 30.3\% \\
7 & 200 & 65.7 \% & 71.7 \% & 74.9\% & 39.1 \% & 47.2\% & 58.1\% \\
7 & 500 &  96.4 \% & 97.5\% & 99.8 \% & 81.3\% & 88.6\% & 90.2\% \\
\hline
10 & 100 & 37.4\% & 43.6\% & 55.2\% & 37.4\% & 44.8\% & 57.2\%  \\
10 & 200 & 77.8\% & 82.4\% & 91.7 \%  & 77.8\% & 82.4\% & 90.9\% \\
10 & 500 & 100\% & 100\% & 100\% & 100\% & 100\% & 100\% \\
\hline
\end{tabular}
\end{center}
\end{table}
}

\begin{table}[t]
\begin{center}
\caption{\it Time (in seconds) taken by different methods of testing separability.
\label{tab9}}
\label{Time}
\bigskip
\begin{tabular}{|c|c|c|c|}
\hline
Method & N =100 & N=200 & N= 500 \\
\hline
M1 ($\psi_1$) & 3.76 & 6.21 & 11.73\\
M1 ($\psi_2$) & 3.89 & 6.75 & 12.01\\
M1 ($\psi_3$) & 3.85 & 6.58 & 11.86\\
\hline
M2 ($\psi_1$) & 11.44 & 23.65 & 58.77 \\
M2 ($\psi_2$) & 12.75 & 28.64 & 62.37\\
M2 ($\psi_3$) & 12.59 & 27.98 & 61.75\\
\hline
M3 ($k=2$) & 2.45 & 4.25 & 9.69\\
M3 ($k=3$) & 4.61 & 8.56 & 20.45\\
M3 ($k=4$) & 11.57 & 22.13 & 52.92\\
\hline
M4 ($L=J=2$) & 3.01 & 5.86 & 9.89\\
M4 ($L=J=3$) & 6.48 & 13.77 & 24.96\\
M4 ($L=J=4$) & 14.32 & 28.64 & 62.35\\
\hline
\end{tabular}
\end{center}
\end{table}

\section{Finite Sample Properties} \label{sec4}
\def\theequation{4.\arabic{equation}}
\setcounter{equation}{0}
In this section we study the finite sample properties of a {family of tests} for the hypothesis of separability described in Section \ref{sec:test} by means of a small simulation study. We also compare the new tests with the tests proposed by
\citet{Aston2017} and \citet{conkokrei2017}
and illustrate potential applications with a data example.
For this purpose we have  implemented  {the asymptotic  test \eqref{testasy} based on simulated quantiles of the  random variable appearing in \eqref{eq:null_dist}, the new  bootstrap test as described in Section \ref{sec:test}} as well as the studentized bootstrap test described in \citet{Aston2017} and the weighted $\chi^2$ test based on the test statistic $\widehat{T}_F$ as described in Theorem 3 of \citet{conkokrei2017}. 
The new test depends on the choice of the operator $\Delta$ and we demonstrate in a simulation study that the test is not very sensitive with respect to this choice. For this purpose we consider  three  integral operators with the following kernels 
\begin{align}
\psi_1(t,t') &\equiv 1~,~ 
\psi_2(t,t') = \vert t - t' \vert, \nonumber\\
\psi_3(t,t') & = \exp(-\pi(t^2+t'^2)).\nonumber
\end{align}
The test due to \citet{Aston2017} requires specification of eigen-subspace, which was chosen to be $I_k=\{(i,j):i=1,\dots,k;j=1,\dots,k\}$ for $k=2,3,4$  and p-values were obtained by empirical non-studentized bootstrap. We used the R package "covsep" [see \citet{covsep}] to implement their method. For the tests proposed by \citet{conkokrei2017} we choose the procedure based on the statistic $\widehat{T}_F$ as in a simulation study it turned out to be the most powerful procedure among the four methods proposed in this paper. The test requires the specifications of the number of spatial  and temporal principle components, which  were both taken to be equal and the number was chosen to be $2, 3$ and $4$. 

\subsection{Simulation Studies}
\label{sec41}
The data were generated from  {a zero-mean Gaussian distribution and a t-distribution with 5 degrees of freedom} with  two different covariance kernels. The first is the spatio-temporal covariance kernel 
\begin{equation}
\label{eq:ker_Gneitning}
c(s,t,s',t') = \frac{\sigma^2}{(a\vert t - t' \vert^{2\alpha} + 1)^{\tau}}\exp\left(- \frac{c\|s - s'\|^{2\gamma}}{(a\vert t - t' \vert^{2\alpha} + 1)^{\beta \gamma}}\right),
\end{equation}
introduced by \citet{gneiting2002nonseparable}.
In this covariance function, $a$ and $c$ are nonnegative scaling parameters of time and space, respectively; $\alpha$ and $\gamma$ are smoothness parameters which take values in the interval $(0, 1]$; $\beta$ is the separability parameter which varies in the interval $[0, 1]$; $\sigma^2 > 0$ is the point-wise variance; and $\tau \geq  \beta d/2$, where $d$ is the spatial dimension. If $\beta = 0$, the covariance is separable and the space-time interaction becomes stronger with increasing values of $\beta$. We fix $\gamma =1, \alpha = 1/2, \sigma^2 = 1, a=1, c=1$ and $\tau = 1$ in the following discussion and choose different values for the parameter
$\beta$ specifying the level of separability.

As a second  example we consider covariance structure
\begin{align}
\label{eq:ker_ch}
c(s,t,s',t') = &\frac{\sigma^2c_0^{d/2}}{(a_0^2(t-t')^2 +1)^{1/2}\Big(a_0^2(t-t')^2 + c_0\Big)^{d/2}} \\
&~~~~~~~~~\times \exp\Big [-b_0\|s - s'\| \Big(\frac{a_0^2(t-t')^2 + 1}{a_0^2(t-t')^2 + c_0}\Big)^{1/2} \Big],\nonumber
\end{align}
which was introduced by \citet{cressie1999classes}.
Here $a_0$ and $b_0$ are non negative scaling parameters of time and space respectively; $c_0 > 0$ is the separability parameter, $\sigma^2> 0$ is the point-wise variance and $d$ is the spatial dimension. The covariance is separable if $c_0 = 1$. For the simulation study we take $a_0 = 2, b_0 =1, \sigma^2 = 1$, $d=2$  and  consider different values of the parameter $c_0$.

We generate data at $100$ equally spaced time points in $[0,1]$ and $11$ space points on the grid $[0,1] \times [0,1]$. The integrals are approximated by average of the function value at grid points.  The nominal significance level is taken to be $5\%$ and empirical rejection region are computed based on $1000$ Monte-Carlo replications and $1000$ bootstrap samples. {In order to estimate the quantiles the asymptotic test \eqref{testasy}  we use
 $1000$ simulation runs.}

{
The simulated rejection probabilities of all four tests under consideration are displayed in 
Table \ref{Gneitning_Gauss} - \ref{tab8}, where the  results for the bootstrap test proposed in this paper (denoted by M1), the  asymptotic test \eqref{testasy} using the simulated quantiles of the limiting null distribution (M2) can be found
in Table 
\ref{Gneitning_Gauss}, \ref{Gneitning_t}, \ref{CH_Gauss}.
and  \ref{CH_t}. For the sake of comparison we show in Table 
\ref{tab2}, \ref{tab4}, \ref{tab6} and \ref{tab8} the
corresponding results for the test using  projections on sub-spaces suggested in \citet{Aston2017} (M3) and the weighted $\chi^2$ test proposed in  \citet{conkokrei2017} (M4).  The computing times of the different procedures are compared in Table \ref{tab9}. The computing times are based on data with 10 equally spaced time and space points both from $[0,1]$ generated from zero-mean Gaussian process with covariance kernel given by \eqref{eq:ker_Gneitning}.}

 {
All procedures yield rather similar results under the null hypothesis, and in general  the nominal level is very well approximated by all tests under consideration. On the other hand, under the alternative we observe more differences. The  asymptotic test \eqref{testasy} proposed in this paper  yields the best power.
The power of the test  of \citet{Aston2017} 
is increasing with the number $k$ of  eigen-subspaces used in the procedure. This  improvement is achieved at the expense of the computing time [see Table \ref{tab9}]. A similar observation can be made for  the test of \citet{conkokrei2017} with respect to the number
of spatial and temporal principle components.  
}

 {
The results of the   bootstrap test proposed in this paper  and the test  of \citet{Aston2017} are similar if the latter is used with $k=2$ sub-spaces, but 
 the test  of \citet{Aston2017}   is more powerful
for $k=3,4$. The bootstrap test proposed in this paper
has more power than the test of \citet{conkokrei2017}
with $L=J=2$ or $L=J=3$ 
spatial and temporal principle components.
For the choice  $L=J=4$ the test 
\citet{conkokrei2017} shows a slightly better performance.}

 {In general an improvement in power 
is always achieved at a cost of computational time  - see Table \ref{tab9}. For example, 
the  asymptotic test \eqref{testasy} proposed in this paper  turns out to be the most powerful procedure, but also requires the most computational time. However, if one wants to improve the power of  the competing tests of \citet{Aston2017}  and \citet{conkokrei2017}, one has to use a larger number of sub-spaces or spatial and temporal principle components, respectively, and this also increases the computational time substantially. 
From Table \ref{tab9} we see that the computational times
of the tests of \citet{Aston2017}  (with $k=4$) and \citet{conkokrei2017}  (with $L=J=4$) are very similar to those of the asymptotic test \eqref{testasy} proposed in this paper. However, the new asymptotic  test is still more powerful.

 {In general more power 
is always achieved at the expense of computational time - see Table \ref{tab9}. For example, 
the  asymptotic test \eqref{testasy} proposed in this paper  turns out to be the most powerful procedure, but also requires the most computational time. However, if one wants to improve the power of  the competing tests of \citet{Aston2017}  and \citet{conkokrei2017}, one has to use a larger number of sub-spaces or spatial and temporal principle components, respectively, and this also increases the computational time substantially. 
From Table \ref{tab9} we see that the computational times
of the tests of \citet{Aston2017}  (with $k=4$) and \citet{conkokrei2017}  (with $L=J=4$) are very similar to those of the asymptotic test \eqref{testasy} proposed in this paper. However the  new test is still more powerful. On the other hand 
the test of \citet{Aston2017} with $k=2$ sub-spaces and the  new bootstrap test seem to provide a balance between good power and reasonable computational time.}  

 {Finally,  the asymptotic  test and the bootstrap test proposed in this paper seem to be very robust with respect to different choices for the kernel $\psi$.    
}
}

\subsection{Application to Real Data}

{
We apply our new methods to the acoustic phonetic data discussed in \citet{Aston2017}. This data set has been compiled in the Phonetics Laboratory
of the University of Oxford between 2012-2013. It consists of natural speech recordings
of five languages: French, Italian, Portuguese, American Spanish and Castilian Spanish. The speakers utter the numbers one to ten in their native language. The data
set consists of a sample of $219$ recordings by $23$ speakers. The More information about this data and related project can be found on the website \url{http://www.phon.ox.ac.uk/ancient_sounds}. The data war first transformed similarly as \citet{Aston2017}, i.e., it was transformed to a smoothed log-spectograms through a short-time Fourier transformation using a Gaussian window function with window-size $10$ miliseconds. The log-spectograms were demeaned separately for each language. We employed the bootstrap test with $B=1000$ (M1) and the test based on the simulated quantiles of the asymptotic distribution of $N\widehat{D}_N$ as appeared in Theorem \ref{thm:null_dist} (M2) on the dataset for three choices of kernels as mentioned in \ref{sec41}. The results are presented in Table \ref{phonetic}. The hypothesis of  separability is clearly rejected.
}

\begin{table}[h]
\begin{center}
\caption{\label{phonetic} {P-values of the  tests (with  different kernels) for the  Phonetic Acoustic Data}}
\begin{tabular}{|c|c|c|c|c|c|c|}
\hline
Languages & M1 ($\psi_1 $) & M1($\psi_2  $) & M1($\psi_3$) &  M2 ($\psi_1 $) & M2($\psi_2  $) & M2($\psi_3$)\\
\hline
French & 0.003 & 0.005 & 0.002 &  0.002 & < 0.001 & <0.001\\
\hline
Italian & <0.001 & <0.001 & <0.001 & < 0.001 & < 0.001 & <0.001\\
\hline
Portuguese & <0.001 & 0.002 & <0.001 & < 0.001 & < 0.001 & <0.001\\
\hline
American Spanish & 0.001 & <0.001 & <0.001 & < 0.001 & < 0.001 & <0.001\\
\hline
Castilian Spanish & <0.001 & <0.001 & <0.001 & < 0.001 & < 0.001 & <0.001 \\
\hline
\end{tabular}
\end{center}
\end{table}

\bigskip
\bigskip

\noindent 	
{\bf Acknowledgments.}
This work has been supported in part by the Collaborative Research Center ``Statistical modelling of nonlinear
dynamic processes'' (SFB 823, Teilprojekt  A1, C1) of the German Research Foundation (DFG).  The authors are grateful to Juan Cuesta Albertos, Shahin Tavakoli for very helpful discussions
and to Martina Stein, who typed parts of this paper with considerable technical expertise. The authors would also like to thank three anonymous referees and the Associate Editor for their constructive comments on an earlier version of this paper.

 \bibliographystyle{apalike}
\bibliography{References}

\begin{appendix}
\section{Technical results and proofs}
\label{sec5}
\def\theequation{A.\arabic{equation}}
\setcounter{equation}{0}
\subsection{Properties of Tensor Product Hilbert Spaces}
\label{a1}
In this section we present some auxiliary  results 
which have been used in the proofs of the main results. The following two results are generalizations of Lemma 1.6 and Lemma 1.7 in the online supplement of \citep{Aston2017}. 
\medskip

\begin{lemma}
\label{lem:D0}
The set  \begin{equation} \label{D0a} \mathcal{D}_0 := \Big\{\sum_{i=1}^n A_i \tilde{\otimes} B_i: A_i \in S_2(H_1), B_i \in S_2(H_2), \text{ are finite rank operators }, n \in \mathbb{N}\Big\}
\end{equation}
is dense in  $S_2(H_1 \otimes H_2)$.
\end{lemma}

\begin{proof}
Let $T \in S_2(H_1 \otimes H_2)$, then $T = \sum_{j\ge 1} \lambda_j u_j \otimes_o v_j$ where $(\lambda_j)_{j \ge 1}$ are the singular values of $T$, with $\sum_j \lambda_j^2 < \infty$; and $\{u_j\}_{j \ge 1}$ and $\{v_j\}_{j \ge 1}$ are orthonormal subsets of $H_1 \otimes H_2$. The vectors $u_j$ and $v_j$'s can be further decomposed as $u_j = \sum_{l \ge 1}s_{j,l}e^{(1)}_l \otimes e^{(2)}_l$ and $v_j = \sum_{l \ge 1}t_{j,l}f^{(1)}_l \otimes f^{(2)}_l$, {where $\{e_l^{(1)}\}_{l \geq 1}$ and $\{f_l^{(1)}\}_{l \geq 1}$ are orthonormal basis of $H_1$; $\{e_l^{(2)}\}_{l \geq 1}$ and $\{f_l^{(2)}\}_{l \geq 1}$ are orthonormal basis of $H_2$ and the series converge in the norm of $H_1 \otimes H_2$.} Let
\begin{align*}
u_j^{M} = \sum_{l=1}^M s_{j,l}e^{(1)}_l \otimes e^{(2)}_l ~~~\text{ and }~~~ v_j^M = \sum_{l =1}^Mt_{j,l}f^{(1)}_l \otimes f^{(2)}_l.
\end{align*}
Given $1 > \epsilon > 0$, choose $N$ large enough such that $\VERT T - \sum_{j=1}^N \lambda_j u_j \otimes_o v_j \VERT_2 \leq \epsilon/2.$ With this we write
\begin{align*}
\left\VERT T - \sum_{j=1}^N \lambda_j u^M_j \otimes_o v^M_j \right\VERT_2 \leq &\left\VERT  T - \sum_{j=1}^N \lambda_j u_j \otimes_o v_j \right\VERT_2\\
&+ \left \VERT \sum_{j=1}^N\lambda_j\left[(u_j - u_j^M)\otimes_o v_j + u_j^M \otimes_o (v_j - v_j^M)\right] \right\VERT_2\\
\leq &~ \epsilon/2 + \sum_{j=1}^N \lambda_j \left\VERT (u_j - u_j^M)\otimes_o v_j + u_j^M \otimes_o (v_j - v_j^M)\right\VERT_2\\
\leq &~ \epsilon/2 + \left[\sum_{j=1}^N \lambda_j^2\right]^{1/2}\left[\sum_{j=1}^N \left\VERT (u_j - u_j^M)\otimes_o v_j + u_j^M \otimes_o (v_j - v_j^M)\right\VERT_2 \right]^{1/2}\\
\leq &~ \epsilon/2 + \VERT T \VERT_2 \left[\sum_{j=1}^N \left[\left\| (u_j - u_j^M)\right\|\left\| v_j\right\| + \left\|u_j^M \right\| \left\|(v_j - v_j^M)\right\|\right] \right]^{1/2}.
\end{align*}
Choose $M \geq 1$ such that
$$\| u_j - u_j^M \| \leq \min\left\{\frac{\epsilon^2}{12N \VERT T \VERT_2^2},1\right\} ~~\text{ and }~~ \| v_j - v_j^M \| \leq \frac{\epsilon^2}{12N \VERT T \VERT_2^2}.$$
As $\|u_j \| = 1$ and $\| u_j^M \| \leq \| u_j \| + \|u_j - u_j^M\|$, with this choice we have
$$\left\VERT T - \sum_{j=1}^N \lambda_j u^M_j \otimes_o v^M_j \right\VERT_2 \leq \epsilon$$
{The proof is now complete by noting that
(under the null hypothesis of separability)
$$\sum_{j=1}^N \lambda_j u^M_j \otimes_o v^M_j = \sum_{j=1}^N \sum_{k,l =1}^M \left(\lambda_j s_{j,l} t_{j,k} e_{l}^{(1)} \otimes f_k^{(1)}\right) \widetilde{\otimes} \left(e_{l}^{(2)} \otimes f_k^{(2)}\right).$$ }
\end{proof}
\medskip

\begin{lemma}
\label{lem:ker_approx}
Let $C \in S_2(L^2(S \times T))$, where $T$ and $S$ are compact subsets of ~$\R^p$ and ~$\R^q$ respectively, be an integral operator with symmetric continuous kernel $c$. For any $\epsilon >0$, there exists an 
operator
$C^{\prime} = \sum_{n=1}^N A_n \tilde{\otimes} B_n$, where $A_n: L^2(S) \to L^2(S)$, $B_n: L^2(T) \to L^2(T)$ are finite rank operators with continuous kernels $a_n$ and $b_n$ respectively, such that
\begin{itemize}
\item [(a)] $\VERT C - C^{\prime} \VERT_2 \leq \epsilon,$
\item [(b)] $\sup_{s,s' \in S, t,t' \in T} \vert c(s,t,s',t') - c^{\prime}(s,t,s',t')\vert \leq \epsilon$, where $c^{\prime}$ is the kernel of the operator $C^{\prime}$.
\end{itemize}
\end{lemma}

\begin{proof}
By Mercer's Theorem {(Proposition 1.2 from \citep{brislawn1988kernels})}, there exists continuous orthonormal functions $\{u_n\}_{n \ge 1} \subset L^2(S \times T)$ and $\lambda_n$ is the summable sequence of positive eigenvalues, such that
$$c(s,t,s',t') = \sum_{n \ge 1} \lambda_n u_n(s,t)u_n(s',t')$$
where the convergence is absolute and uniform.

Let $U_n$ be the integral operator with kernel $u_n$ and $C_N := \sum_{n=1}^N \lambda_n U_n \otimes_o U_n$. Denote the kernel of $C_N$ as $c_N$. As $u_n$  can be approximated by sums of tensor products of continuous functions, let ${ u_n^M(s,t)} := \sum_{l=1}^M f_{n,l}^{(1)}(s) f_{n,l}^{(2)}(t)$, where $f_{n,l}^{(1)} \in L^2(S), f_{n,l}^{(2)} \in L^2(T)$ are continuous, and choose $M$ such that
{ $$\int_{T}\int_{S} \sup_{s,t} \vert u_n(s,t) - u_n^{M}(s,t)\vert dsdt \leq \min\left\{\frac{\epsilon^2}{12N\kappa\VERT C \VERT_2^2},\kappa\right\},$$}
where $\kappa = \max_{n=1, \dots,N}\int_T\int_S \sup_{s, t} \vert u_n(s,t)\vert dsdt.$ Writing the kernel $$c^{N,M}(s,t,s',t') := \sum_{n=1}^N \lambda_n u_n^M(s,t)  U_n^M(s',t'),$$ we have
\begin{align*}
& \int_T\int_T\int_S\int_S \sup_{s,t,s',t'} \vert c_N(s,t,s',t') - c^{N,M}(s,t,s't')\vert dsds'dtdt'\\
&\leq \sum_{n=1}^N \lambda_n\left[\int_{T}\int_{S} \sup_{s,t} \vert u_n(s,t) - u_n^{M}(s,t)\vert dsdt\right]\left[\int_{T}\int_{S} \sup_{s,t} \left(\vert u_n(s,t) \vert + \vert u_n^{M}(s,t)\vert\right) dsdt\right].
\end{align*}
An application of the Cauchy-Schwartz inequality along with the choice of $M$ gives an upper bound to the last quantity to be $\epsilon/2$. Similar calculations show that 
$$
\int_T\int_T\int_S\int_S \sup_{s,t,s',t'} \vert c_N(s,t,s',t') - c^{N,M}(s,t,s't')\vert dsds'dtdt' \leq \epsilon/2.
$$
Finally, as $c^{N,M}$ is indeed a finite sum of tensor products of finite rank kernels, we have the desired result.
\end{proof}
\medskip

We conclude this section with a simple result about Gaussian processes on a Hilbert space. For this purpose  recall that a
random element $\mathcal{G} $ on  a real separable Hilbert space
$H$ is said to be Gaussian with mean $\mu \in H$ and covariance operator $\Gamma : H \mapsto H$ if for all $x \in H$, the random variable $\langle \mathcal{G}, x \rangle$ has a normal distribution with mean $\langle \mu, x \rangle$ and variance $\langle \Gamma x, x \rangle$  (See Section 1.3 of \citet{lifshits2012lectures} for more details).

\begin{lemma}
\label{lem:gauss_ten}
Let $H_1$ and $H_2$ be two real separable Hilbert spaces and $\mathcal{G}$ be a Gaussian process on $S_2(H_2)$. Then for all $A \in S_2(H_1)$, the process $A \tilde{\otimes} \mathcal{G}$ is a Gaussian process in $S_2(H_1 \otimes H_2)$.
\end{lemma}

\begin{proof}
We will show for any $T \in S_2(H_1 \otimes H_2)$, the random variable $\langle T, A \tilde{\otimes} \mathcal{G} \rangle$ has a normal distribution. By Lemma \ref{lem:D0} and the  continuity of the inner product it is enough to show the result for $T \in \mathcal{D}_0$. Therefore let $T = \sum_{n=1}^N A_n \tilde{\otimes} B_n$ then $$\langle T, A \tilde{\otimes} \mathcal{G} \rangle = \sum_{n=1}^N \langle A_n,A \rangle_{S_2(H_1)} \langle B_n, \mathcal{G} \rangle_{S_2(H_2)}$$
which is sum of normal random variables and hence normal.
\end{proof}

\subsection{Proof of Proposition \ref{prop:t1}}
\label{a2}
By Lemma \ref{lem:D0} in Section  \ref{a1}
 the space
\begin{equation} \label{D0} \mathcal{D}_0 := \Big\{\sum_{i=1}^n A_i \tilde{\otimes} B_i: A_i \in S_2(H_1), B_i \in S_2(H_2), n \in \mathbb{N}\Big\}
\end{equation}
is a dense subset of $S_2(H_1 \otimes H_2)$ (note that a similar result for the space $S_1(H_1 \otimes H_2)$ has been established in Lemma 1.6 of the supplementary material in \citet{Aston2017}).  For all $B \in S_2(H_2)$, $E \in \mathcal{D}_0$ and $C_1 \in S_2(H_1)$, we have
\begin{align}
\langle B, T_1(E,C_1) \rangle_{S_2(H_2)}= & \Big\langle B, \sum_{i=1}^n \langle A_i,C_1\rangle_{S_2(H_1)} B_i\Big\rangle_{S_2(H_2)}
= \sum_{i=1}^n \langle  A_i,C_1 \rangle_{S_2(H_1)} \langle B,B_i \rangle_{S_2(H_2)}\nonumber\\
= &\Big \langle \sum_{i=1}^n A_i \tilde{\otimes} B_i , C_1 \tilde{\otimes} B\Big\rangle_{HS} = \langle E, C_1 \tilde{\otimes} B  \rangle_{HS}.\label{eq:t1d0}
\end{align}
Using the fact that 
\begin{align} \label{test}
\VERT T_1(E,C_1) \VERT_2 \leq \VERT T_1(E,C_1) \VERT_1 =\sup\big\{\langle B, T_1(E,C_1)\rangle_{S_2(H)} : B \in S_2(H_2), \VERT B \VERT_{\infty} \leq 1\big\},
\end{align}
\eqref{eq:t1d0} and the Cauchy Schwarz inequality it follows 
that 
\begin{equation}\label{new1}
 \VERT T_1(E,C_1) \VERT_2 \leq \VERT C_1 \VERT_2 \VERT E \VERT_2. 
\end{equation}
Therefore, for each $C_1 \in S_2(H_1)$, we can extend $T_1(.,C_1)$ continuously on $S_2(H)$.

Furthermore
as \eqref{T1} holds for all $C \in \mathcal{D}_0$ and the maps $C \mapsto \langle B, T_1(C,C_1) \rangle_{S_2(H_1)}$ and $C \mapsto \langle C, C_1 \tilde{\otimes} B  \rangle_{HS}$ are continuous for all $B \in S_2(H_2)$ and $C_1 \in S_2(H_1)$,  we can conclude that \eqref{T1} holds for all $C \in S_2(H)$.

\subsection{Proof of Proposition \ref{prop:t1alt}}
\label{a3}

Recall the definition of the set
$\mathcal{D}_0$ in \eqref{D0} and let
$C = \sum_{n=1}^N A_n \tilde{\otimes} B_n \in \mathcal{D}_0$, where $A_n \in S_2(H_1)$, $B_n \in S_2(H_2)$. We write
\begin{align*}
\langle C, C_1 \tilde{\otimes} T_1(C,C_1)\rangle_{HS} = &\Big \langle C, C_1 \tilde{\otimes} \sum^N_{n=1} \langle A_n,C_1 \rangle_{S_2(H_1)} B_n \Big \rangle_{HS}\\
= & \sum^N_{n=1}\langle A_n,C_1 \rangle_{S_2(H_1)}\langle C, C_1 \tilde{\otimes} B_n\rangle_{HS}\\
= &\sum^N_{n=1} \sum^N_{m=1} \langle A_n,C_1 \rangle_{S_2(H_1)} \langle A_m,C_1 \rangle_{S_2(H_1)} \langle B_m,B_n \rangle_{S_2(H_2)}.
\end{align*}
On the other hand,
\begin{align*}
\langle T_1(C,C_1), T_1(C,C_1) \rangle_{S_2(H_2)} = &\Big \langle \sum^N_{n=1} \langle A_n,C_1\rangle_{S_2(H_1)} B_n, \sum^N_{m=1} \langle A_m,C_1\rangle_{S_2(H_1)} B_m \Big \rangle_{S_2(H_2)}\\
= &\sum^N_{n=1} \sum^N_{m=1} \langle A_n,C_1 \rangle_{S_2(H_1)} \langle A_m,C_1 \rangle_{S_2(H_1)} \langle B_m,B_n \rangle_{S_2(H_2)}.
\end{align*}
Therefore, for all $C_1 \in S_2(H_1)$, the functions $f,g: S_2(H) \mapsto \R$ defined by
$$
f(C) := \langle C, C_1 \tilde{\otimes} T_1(C,C_1)\rangle_{HS} ~~ \mbox{and} ~~ g(C) := \VERT T_1(C,C_1)\VERT_2^2
$$
are continuous and coincide on the dense subset ${\cal D}_0$ of $S_2(H)$. So $f(C) = g(C)$ for all $C \in S_2(H)$ and hence the result follows.

\subsection{Proof of Proposition \ref{prop:t1kernel}}
\label{a4}
By Lemma \ref{lem:ker_approx} for a given $\epsilon>0$ there exists an integral operator $C'$ with kernel $c'$ such that
$$
\VERT C - C' \VERT_2 < \frac{\epsilon}{2\VERT C_1 \VERT_2} ~\mbox{ and } ~\|c-c'\|_{\infty} < \epsilon/2,
$$
where $C' = \sum_{n=1}^N A_n \tilde{\otimes} B_n$, and $A_n$ and $B_n$ are finite rank operators with continuous kernels $a_n$ and $b_n$.
Note that $\sum_{n=1}^N a_n b_n$  is the kernel of the operator  $C^\prime $.
Let $K$ be the integral operator with the kernel defined in \eqref{kdef}
and $K'$ be the integral operator with kernel 
$$
k'(t,t') := \int_S \int_{S} c'(s,t,s't')c_1(s,s')ds ds',
$$
then (note that $K'$ is a Hilbert-Schmidt operator)
\begin{align}
\VERT T_1(C,C_1) - K \VERT_2 \leq & \VERT T_1(C,C_1) - T_1(C',C_1)\VERT_2 + \VERT T_1(C',C_1) - K'\VERT_2 + \VERT K'- K \VERT_2. \label{ineq:T1}
\end{align}
By  \eqref{new1} the first term is bounded by $\VERT C_1 \VERT_2 \VERT C - C_1 \VERT_2 < \epsilon/2$. To handle the second term note that for any $f \in H_2$,
\begin{align*}
T_1(C',C_1)f(t) = &T_1\left(\sum_{n=1}^N A_n \tilde{\otimes} B_n, C_1\right)f(t) = \sum_{n=1}^N \langle A_n,C_1 \rangle_{S_2(H_1)}B_nf(t)\\
= &\sum_{n=1}^N \int_T\int_S \int_S a_n(s,s')c_1(s,s')b_n(t ,t')f(t')dsds'dt'\\
= &\int_T \int_S \int_S c'(s,t,s',t')c_1(s,s')dsds'f(t')dt'\\
= & \int_T k'(t,t')f(t')dt' = K'f(t).
\end{align*}
Therefore, the second term in \eqref{ineq:T1} is zero. If $|S|$ and $|T|$ denote the Lebesgue measure of the set $S$ and $T$, respectively, the third term  can be written as
\begin{align*}
&\left[\int_T\int_T\left(\int_S\int_S(c(s,t,s',t') -c'(s,t,s't'))c_1(s,s')dsds'\right)^2dtdt'\right]^{1/2}\\
&~~~~~~~~~~~~~~~~~~~~~~~~~~~~~~~~~~~~~~~~~~\leq \|c - c'\|_{\infty}\|c_1\|_2 \vert S \vert^2 \vert T \vert,
\end{align*}
which is bounded  by $\|c_1\|_2 \vert S \vert^2 \vert T \vert\epsilon/2$, {as $S$ and $T$ are bounded}. Since the choice of $\epsilon>0$ is arbitrary, this proves the assertion of Proposition \ref{prop:t1kernel}.

{
\subsection{Calculation of the limiting variance in Theorem \ref{thm:asym_dist}\label{a6}}
{
For the calculation of the limiting variance, we first calculate the 
elements of  the covariance matrix $\Sigma$ starting by simplifying the second coordinate of $F_2(\mathcal{G},C)$, that  is 
\begin{align*}
\Big\langle C,T_2(\mathcal{G},\Delta) \widetilde{\otimes} T_1(C,T_2(C,\Delta))\Big\rangle_{HS} = &\Big\langle T_2(\mathcal{G},\Delta), T_2(C, T_1(C,T_2(C,\Delta)))\Big\rangle_{S_2(H_1)}\\
&=\Big\langle \mathcal{G}, T_2(C, T_1(C,T_2(C,\Delta)))\widetilde{\otimes} \Delta \Big\rangle_{HS}~,
\end{align*}
where we used  Proposition \ref{prop:t1}.  
With the notations 
$$
T_2^{\Delta} := T_2(C,\Delta)~,~ ~ T_{12}^{\Delta} := T_1(C,T_2(C,\Delta))~,~~ 
T_{212}^{\Delta} := T_2(C,T_1(C,T_2(C,\Delta)))
$$ 
the elements in the covariance matrix $\Sigma = (\Sigma_{ij})_{i,j=1}^3$ of the random vector 
\begin{align*}
F_2(\mathcal{G},C) = 2\left(
\begin{array}{c} 
\left\langle \mathcal{G},C\right\rangle_{HS}\\
\left\langle \mathcal{G},  T_2^{\Delta} \widetilde{\otimes}  T_{12}^{\Delta}\right\rangle_{HS} + \langle \mathcal{G}, T_{212}^{\Delta}\widetilde{\otimes} \Delta \rangle_{HS} \\
\langle \mathcal{G},T_2^{\Delta} \widetilde{\otimes} \Delta \rangle_{HS}
\end{array}
\right)~
\end{align*}
are  given by
\begin{align*}
\Sigma_{11} = &4\Big\langle \Gamma C, C \Big\rangle_{HS}\\
\Sigma_{22} = &4\Big\langle \Gamma T_2^{\Delta}\widetilde{\otimes}T_{12}^{\Delta}, T_2^{\Delta}\widetilde{\otimes}T_{12}^{\Delta} \Big\rangle_{HS} + 4\Big\langle \Gamma T_{212}^{\Delta}\tilde{\otimes } \Delta,T_{212}^{\Delta}\tilde{\otimes } \Delta \Big\rangle_{HS}
+  8\Big\langle \Gamma T_2^{\Delta}\widetilde{\otimes}T_{12}^{\Delta},T_{212}^{\Delta}\tilde{\otimes } \Delta \Big\rangle_{HS}\\
\Sigma_{33} = &4\Big\langle \Gamma T_2^{\Delta}\tilde{\otimes } \Delta,T_2^{\Delta}\tilde{\otimes } \Delta \Big\rangle_{HS}\\
\Sigma_{12} = &4\Big\langle \Gamma C,  T_2^{\Delta}\widetilde{\otimes}T_{12}^{\Delta} \Big\rangle_{HS}+ 4\Big\langle \Gamma C,T_{212}^{\Delta}\tilde{\otimes } \Delta \Big\rangle_{HS}\\
\Sigma_{13} = &4\Big\langle \Gamma C,T_2^{\Delta}\tilde{\otimes } \Delta\Big\rangle_{HS}\\
\Sigma_{23} = &4\Big\langle \Gamma T_2^{\Delta}\widetilde{\otimes}T_{12}^{\Delta},T_2^{\Delta}\tilde{\otimes } \Delta\Big\rangle_{HS} + 4\Big\langle \Gamma T_{212}^{\Delta}\tilde{\otimes } \Delta,T_{2}^{\Delta}\tilde{\otimes } \Delta \Big\rangle_{HS}.
\end{align*}
A straightforward calculation then gives 
\begin{align*}
\nu^2 := &4\langle \Gamma C, C\rangle_{HS} + 4\frac{\left\langle \Gamma \big(T_2^{\Delta} \widetilde{\otimes} T_{12}^{\Delta}\big),T_2^{\Delta} \widetilde{\otimes} T_{12}^{\Delta}\right\rangle_{HS}}{\VERT T_2^{\Delta}\VERT_2^4} + 4\frac{\left\langle \Gamma \big(T_{212}^{\Delta} \widetilde{\otimes} \Delta\big),T_{212}^{\Delta} \widetilde{\otimes} {\Delta}\right\rangle_{HS}}{\VERT T_2^{\Delta}\VERT_2^4}\\
&+8\frac{\left\langle \Gamma \big(T_2^{\Delta} \widetilde{\otimes} T_{12}^{\Delta}\big),T_{212}^{\Delta} \widetilde{\otimes} {\Delta}\right\rangle_{HS}}{\VERT T_2^{\Delta}\VERT_2^4} + \frac{4\VERT T_{12}^{\Delta}\VERT_2^4}{\VERT T_2^{\Delta}\VERT_2^8}{\left\langle \Gamma \big(T_2^{\Delta} \widetilde{\otimes} {\Delta}\big),T_2^{\Delta} \widetilde{\otimes} {\Delta}\right\rangle_{HS}}\\
&- \frac{8}{\VERT T_2^{\Delta}\VERT_2^2}\left\langle \Gamma C,T_{212}^{\Delta} \widetilde{\otimes} {\Delta}\right\rangle_{HS} - \frac{8}{\VERT T_2^{\Delta}\VERT_2^2}\left\langle \Gamma C,T_{2}^{\Delta} \widetilde{\otimes} T_{12}^{\Delta} \right\rangle_{HS} + \frac{8\VERT T_{12}^{\Delta}\VERT_2^2}{\VERT T_2^{\Delta}\VERT_2^4}{\left\langle \Gamma C,T_2^{\Delta} \widetilde{\otimes} {\Delta}\right\rangle_{HS}}\\
& - \frac{8\VERT T_{12}^{\Delta}\VERT_2^2}{\VERT T_2^{\Delta}\VERT_2^6}{\left\langle \Gamma \big(T_2^{\Delta} \widetilde{\otimes} T_{12}^{\Delta}\big),T_2^{\Delta} \widetilde{\otimes} {\Delta}\right\rangle_{HS}}
- \frac{8\VERT T_{12}^{\Delta}\VERT_2^2}{\VERT T_2^{\Delta}\VERT_2^6}{\left\langle \Gamma \big(T_{212}^{\Delta} \widetilde{\otimes} {\Delta}\big),T_2^{\Delta} \widetilde{\otimes} {\Delta}\right\rangle_{HS}}.\\
= &A_1 + A_2 + A_3 + A_4 + A_5 + A_6 + A_7 + A_8 + A_9 + A_{10}, 
\end{align*}
where the last equality defines the terms $A_{1}, \ldots , A_{10}$ in an obvious manner.
Observing the definition of $A$ and $B$ in Theorem \ref{thm:asym_dist} we have 
$$A_1+A_2 + A_7 = 4\langle \Gamma A, A \rangle_{HS}~~~~~~~A_3 + A_5 + A_{10} = 4\langle \Gamma B, B \rangle_{HS}$$
$$A_4 + A_6 + A_8 + A_9 = - 2 \langle \Gamma A, B \rangle_{HS} = - 2 \langle \Gamma B, A \rangle_{HS}.
$$
Therefore the limiting variance simplifies to the expression in \eqref{nusquare}.}
}

\end{appendix}

\end{document}